\documentclass{elsarticle}
\usepackage{hyperref}
\journal{IFAC Journal of Systems and Control}

\usepackage{amsmath,amsthm}
\usepackage{amssymb}
\usepackage{graphicx,setspace}
\usepackage{amsmath,amssymb}
\usepackage{color}
\usepackage[english]{babel}
\usepackage{mathrsfs,dsfont}
\usepackage{array}
\usepackage{tikz}
\usepackage{adjustbox}
\usepackage{multirow}
\usetikzlibrary{arrows, automata, decorations.pathmorphing, backgrounds, positioning, fit, petri, matrix, patterns,shapes,snakes}

\usepackage[linesnumbered,vlined,ruled]{algorithm2e}

\newtheorem{theorem}{\bf Theorem}
\newtheorem{lemma}[theorem]{\bf Lemma}

\newcommand{\Na}{\mathbb{N}}

\begin{document}

\begin{frontmatter}

\title{Fleet Sizing and Charger Allocation in Electric Vehicle Sharing Systems}

\author[mymainaddress]{Yuntian Deng\corref{mycorrespondingauthor}}\ead{deng.556@osu.edu}

\author[mymainaddress]{Abhishek Gupta}
\cortext[mycorrespondingauthor]{Corresponding author}
\ead{gupta.706@osu.edu}

\author[mymainaddress,mysecondaryaddress]{Ness B. Shroff}\ead{shroff.11@osu.edu}

\address[mymainaddress]{Department of Electrical and Computer Engineering, The Ohio State University, Columbus, OH 43210, USA}
\address[mysecondaryaddress]{Department of Computer Science and Engineering, The Ohio State University, Columbus, OH 43210, USA}

\begin{abstract}
In this paper, we propose a closed queueing network model for performance analysis of electric vehicle sharing systems with a certain number of chargers in each neighborhood. Depending on the demand distribution, we devise algorithms to compute the optimal fleet size and number of chargers required to maximize profit while maintaining a certain quality of service. We show that the profit is concave with respect to the fleet size and the number of chargers at each charging point. If more chargers are installed within the city, we show that it can not only reduce the fleet size, but it also improves the availability of vehicles at all the points within a city. We further show through simulation that two slow chargers may outperform one fast charger when the variance of charging time becomes relatively large in comparison to the mean charging time.
\end{abstract}


\end{frontmatter}

%\linenumbers

\section{Introduction}
Many governments across the world are emphasizing the decarbonization of transportation to curb greenhouse gas emissions and pollution associated with transportation industry. With the rise in sharing economy within transportation sector, there is a shift from personally-owned modes of transportation to shared vehicles (using ride-hailing services, bikes, and scooters, etc.). These service providers are now commonly referred to as transportation network companies (TNCs). TNCs provide on-demand transportation services to the passengers, increases vehicle utilization, and enhances overall convenience to the passengers. Incorporating electric vehicle within TNCs is widely considered to be a solution to achieve long-term sustainable transportation objectives. 

Given the high annual miles traveled by vehicles in shared fleets, Pavlenko et. al. \cite{pavlenko2019does} estimates that the per mile operating cost of all battery electric or hybrid vehicle fleets is much lower than that of the conventional ones. Even without the current purchasing incentives, long range battery electric vehicles (BEVs) will become the most economically attractive technology for ride-hailing operations by 2023-2025 (assuming the cost of battery packs go down by 35\%). Another recent research suggests that the operating and ownership expenses of electric vehicles with high annual miles driven are significantly lower than those of conventional vehicles \cite{weldon2018long}. However, charging electric vehicles takes a non-trivial amount of time (depending on the state of charge of the vehicle); this was the primary . Moreover, the cost of installing and maintaining charging infrastructure is substantial (installation of a charging station could cost anywhere between \$ 10 - 50 thousand and reserving parking spots for electric vehicles could be costly in high population density areas). Consequently, for the adoption of electric vehicle sharing system, it is important to determine based on the demand distribution:
\begin{enumerate}[(1)]
	\item What should be the optimal fleet size, since a large fleet of EVs results in improved availability, reliability, and better quality of service, but it also costs more to maintain.
	\item What should be an optimal number of charging stations across the city. Again, a large number of charging stations would improve availability and quality of service, but again it has high recurring costs.
\end{enumerate}
We formulate these two problems as optimization problems in which the movement of the vehicles across a city is modeled using a closed queueing network model. We devise algorithms to solve these optimization problems. We also shed some light on the question of nature of charger (fast vs slow) needed for improving the quality of service. 

A significant amount of work has focused on the planning of EV charging infrastructure \cite{shen2019optimization} for personal EVs, such as flow-based models  \cite{hodgson1990flow} and network equilibrium models \cite{he2013optimal}. However, fleet sizing and planning of charging infrastructure for fleets of electric vehicles owned by TNCs have not received much attention.  

Boyaci et al. \cite{boyaci2015optimization} developed and solved a multi-objective mixed-integer programming model for planning one-way vehicle-sharing systems, taking into account vehicle relocation and charging requirements. An aggregate model using the concept of the virtual hub is solved with a branch-and-bound approach. In their model, they assumed that the returned vehicle has to stay in the station for a fixed time, which represents the charging period of the vehicle.

Weikl and Bogenberger \cite{weikl2015practice} introduced a relocation model for free-floating car-sharing systems with both conventional and electric vehicles. Both the refueling of conventional vehicles and the recharging of EVs are taken into account. He et al. \cite{he2017service} studied the trade-offs between maximizing passenger catchment by covering travel needs and controlling fleet operation costs. A distributionally robust optimization framework is developed, which informs robust decisions to overcome possible ambiguity of data, by solving a mixed-integer second-order cone program. 

Queueing network models have now been widely used in the analysis of fleet sizing and vehicle routing problems. George and Xia \cite{george2011fleet} proposed a queuing network based model for a vehicle rental system. Fanti et al. \cite{fanti2014fleet} considered three different types of electric vehicles in a closed queueing network model for EV sharing system: fully charged, partially charged and out of charge. These studies consider the problem of optimal fleet sizing with additional constraints on quality of service parameters such as availability and wait time for the passengers. Our model is largely inspired by these works. 

Zhang and Pavone \cite{zhang2016control} leveraged closed Jackson networks to model an on-demand transportation system, where the rebalancing techniques are applied to a system sizing example for three Manhattan neighborhoods. Zhang et al.  \cite{zhang2018analysis} used two coupled closed Jackson networks with passenger loss to model vehicle rebalancing. Two open-loop control approaches are proposed for system sizing. Iglesia et al. \cite{iglesias2019bcmp} cast an autonomous mobility-on-demand system within the framework of BCMP queueing networks, which can capture both congestion effects and vehicle charging. However, they assume that the charging infrastructure at each site is unlimited, which is not realistic. 

\subsection{Key Contributions of this paper}
The contributions of this paper are threefold. Firstly, we model various neighborhoods in the city as nodes and edges and model the electric vehicle sharing system as a closed BCMP queueing network. Through this modeling framework, we can incorporate stochastic passenger arrival, electric vehicle routing and charging, the general distribution of travel time, the effect of the number of chargers on availability, and the distribution of electric vehicles. Secondly, since different neighborhoods may have different requirements on vehicle availability and may have differing numbers of chargers, we propose optimization formulations to determine the optimal fleet size and charger allocation to maximize profit. Concavity of profit with respect to fleet size and the number of chargers is established. 
We further devise a greedy algorithm based heuristic to solve the charger allocation problem in an iterative fashion.  Thirdly, we show that under a fixed budget,  two slow chargers outperform one fast charger if the variance of charging time becomes relatively large. An approximation method is proposed for general passenger inter-arrival time. Finally, large-scale simulations validate our theorems and approaches.

\subsection{Outline of the paper}
The rest of the paper is organized as follows. The system model is presented in Section \ref{sec:model}. Section \ref{sec:queue} reviews some results on invariant distribution and throughput of closed queueing networks. The optimal fleet sizing problem is formulated and solved in Section \ref{sec:fleet}. In Section \ref{sec:charge}, we devise a greedy approach based heuristic algorithm for charger allocation based on marginal allocation. The comparison between one fast and two slow chargers is discussed in Section \ref{sec:select}. To reduce the computational complexity for large-scale simulations, the mean value analysis is reviewed in Section \ref{sec:algorithm}. We also discuss in this section an approximation approach for computing the stationary probability distribution of the system for the case of general passenger inter-arrival time distribution. Section \ref{sec:numeric} presents the numerical results for both optimization problems for a 60-node city network. Finally, Section \ref{sec:conclusion} concludes the paper and presents directions for future research.

\section{System Modeling}\label{sec:model}
In this section, we model an electric vehicle sharing system using the framework of Baskett, Chandy, Muntz, and Palacios (BCMP) closed queueing network \cite{baskett1975open}. Suppose there are $M \in \Na$ electric vehicles in one city, which can offer services to passengers and are routed around to serve stochastic demands at different places. Similar to \cite{george2011fleet,  iglesias2019bcmp}, from a virtual service view, we model this system as a closed network with respect to vehicles,  i.e., vehicles are routed within this network and receive 'service' at different nodes. The number of vehicles remains constant and there is no vehicle entering or leaving the network.

As shown in Figure \ref{fig:System}, we use three different queues to model three processes: departure, charging and travel. To be specific, each station is mapped into the combination of a single-server queue (SS) node (departure) and a finite-server queue (FS) node (charging), while each travel between stations is mapped into an infinite-serer queue (IS) node. In the following, we explain the system from view of passenger and electric vehicle, respectively.

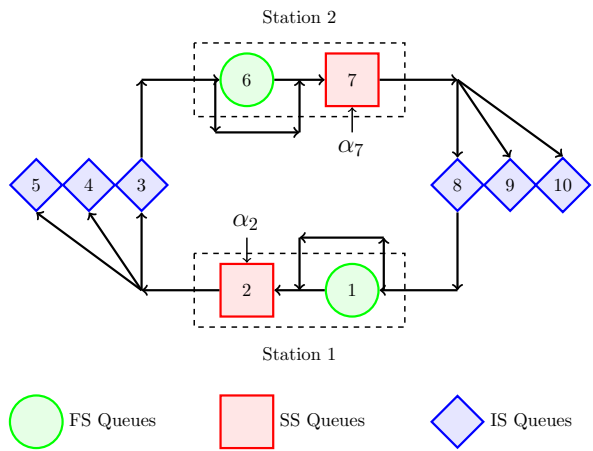
\begin{figure}[bth]
	\centering
	\scalebox{0.7}{\begin{tikzpicture}[
roundnode/.style={circle, draw=green!100, fill=green!10, very thick, minimum size=10mm},
squarednode/.style={rectangle, draw=red!100, fill=red!10, very thick, minimum size=10mm},
diamondnode/.style={diamond, draw=blue!100, fill=blue!10, very thick, minimum size=10mm},
]
%Nodes
\node[roundnode]        (fs2)   at (2,4) {6};
\node[roundnode]        (fs1)   at (4,0) {1};
\node[squarednode]      (ss1)   at (2,0)  {2};
\node[squarednode]      (ss2)   at (4,4)  {7};

\node[diamondnode]      (is1)   at (0,2) {3};
\node[diamondnode]      (is2)   at (-1,2) {4};
\node[diamondnode]      (is3)   at (-2,2) {5};
\node[diamondnode]      (is4)   at (6,2) {8};
\node[diamondnode]      (is5)   at (7,2) {9};
\node[diamondnode]      (is6)   at (8,2) {10};

%Lines
\draw[->][very thick] (ss1.west) -- (0,0);
\draw[->][very thick] (0,0) -- (is1.south);
\draw[->][very thick] (0,0) -- (is2.south);
\draw[->][very thick] (0,0) -- (is3.south);

\draw[->][very thick] (is1.north) -- (0,4);
\draw[->][very thick] (0,4) -- (fs2.west);
\draw[->][very thick] (fs2.east) -- (ss2.west);
\draw[->][very thick] (ss2.east) -- (6,4);

\draw[->][very thick] (6,4) -- (is4.north);
\draw[->][very thick] (6,4) -- (is5.north);
\draw[->][very thick] (6,4) -- (is6.north);

\draw[->][very thick] (is4.south) -- (6,0);
\draw[->][very thick] (6,0) -- (fs1.east);
\draw[->][very thick] (fs1.west) -- (ss1.east);

\draw[->][very thick] (1.4, 4) -- (1.4 ,3);
\draw[->][very thick] (1.4 ,3) -- (3,3);
\draw[->][very thick] (3,3) -- (3,4);

\draw[->][very thick] (4.6, 0) -- (4.6, 1);
\draw[->][very thick] (4.6, 1) -- (3,1);
\draw[->][very thick] (3,1) -- (3,0);

\node[right] at (1.6,1.3) (m) {\Large $\alpha_2$};
\draw[->][thick] (2,1) -- (ss1.north);
\node[right] at (3.6,2.7) (m) {\Large $\alpha_7$};
\draw[->][thick] (4,3) -- (ss2.south);

\node[roundnode] at (-2, -2.5)  {};
\node[right] at (-1.5,-2.5) (m) {FS Queues};

\node[squarednode] at (2,-2.5) {};
\node[right] at (2.5,-2.5) (m) {SS Queues};

\node[diamondnode] at (6,-2.5) {};
\node[right] at (6.5,-2.5) (m) {IS Queues};

\draw[dashed,thick] (1,0.7) rectangle (5,-0.7);
\node at (3,-1.2) (m) {Station 1};

\draw[dashed,thick] (1,4.7) rectangle (5,3.3);
\node at (3,5.2) (m) {Station 2};
\end{tikzpicture}}
	\caption{\label{fig:System} Electric Vehicle Sharing System, in which finite server queues denote charging stations, single server queues denote the passenger pickup stations, and infinite server queues represent road networks.}
\end{figure}    

\subsection{Passenger} 
Assume that there are several stations within a certain geographical area. Each station has a departure point (pick-up node) and a charging point (drop-off node). We denote the set of departure points as $S$ and the set of charging points as $F$. At each departure point $i \in S$, passengers arrive  according to a Poisson arrival with rate $\alpha_i > 0$. If there is at least one electric vehicle waiting at departure point $i$, the passenger will take the first electric vehicle in line and start traveling without any waiting time. 

We assume that there is a "passenger loss" if there is no electric vehicle at the departure point $i$ when passengers arrive, i.e., passengers immediately leave this system and try other transportation systems to finish their trips, rather than forming a queue at node $i$.

Before departing from point $i\in S$, each passenger will select his/her destination as charging point $j\in F$ with probability $p_{ij}$,  where $\sum_{j\in F}p_{ij}=1$ for each $i\in S$. 
According to $p_{ij}$, passengers will have different travel time and we model them as they enter different IS queues.  During the travel from node $i$ and to $j $, we assume that the travel time follows a general distribution with mean $T_{ij}$. %and variance $\sigma_{ij}$.

Passengers exit this system as soon as they arrive at the charging point $j$. For example, if a passenger wants to travel from station 1 to station 2, as shown in Figure \ref{fig:System}, he/she will depart from node $2$ and enter node $3$. Once the passenger leaves node $3$ and arrives at node $6$, then the one-way mobility-on-demand service is finished and the passenger exits this system.

\subsection{Electric Vehicles}
Electric vehicles are routed among three types of queues in the network according to probability $r_{ij}$ (defined later by \eqref{equ:r_ij}). We assume that the transfer from one queue to another is instant. From a virtual service view, electric vehicles form queues and receive services when they are waiting for passengers (SS), traveling between stations (IS), and charging at charging points (FS). Charging can also be skipped according to a certain probability, and in this case, electric vehicles directly go to the departure points after leaving the traveling nodes.

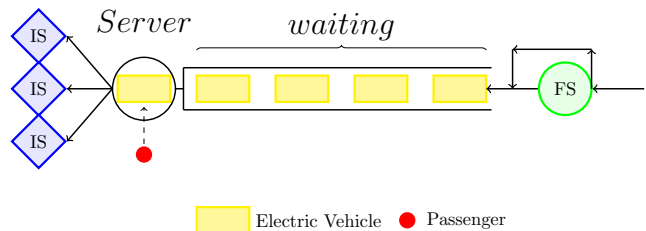
\begin{figure}[bth]
	\centering
	\scalebox{0.7}{\begin{tikzpicture}[
roundnode/.style={circle, draw=green!100, fill=green!10, very thick, minimum size=10mm},
squarednode/.style={rectangle, draw=red!100, fill=red!10, very thick, minimum size=10mm},
diamondnode/.style={diamond, draw=blue!100, fill=blue!10, very thick, minimum size=10mm},
]

%EV
\filldraw[color= yellow!50, draw=yellow, very thick] (-1.5,0) rectangle (-0.5,0.5);

\filldraw[color= yellow!50, draw=yellow, very thick] (0,0) rectangle (1,0.5);
\filldraw[color= yellow!50, draw=yellow, very thick] (1.5,0) rectangle (2.5,0.5);
\filldraw[color= yellow!50, draw=yellow, very thick] (3,0) rectangle (4,0.5);
\filldraw[color= yellow!50, draw=yellow, very thick] (4.5,0) rectangle (5.5,0.5);

%Server
\draw [thick] (-1,0.25) circle (17pt);
\draw [thick] (-0.25,0.25)-- (-0.4,0.25);

%Queue
\draw[thick] (-0.25,-0.15) -- (-0.25,0.65);
\draw[thick] (-0.25,-0.15) -- (5.6,-0.15);
\draw[thick] (-0.25,0.65) -- (5.6,0.65);

%Name
\draw [decoration={brace,mirror,raise=28pt},decorate]
     (5.5,0) --  node[above=30pt]  {\LARGE $waiting$} (0,0);
     
\node [above] at (-1,1.2) {\LARGE $Server$};

%customer
\filldraw[red,thick] (-1,-1) circle (4pt);
\draw[->][dashed] (-1,-0.8) -- (-1,-0.1);

%FS IS
\node[roundnode]        (fs1)   at (7,0.25) {FS};
\node[diamondnode]      (is1)   at (-3,0.25) {IS};
\node[diamondnode]      (is2)   at (-3,-0.75) {IS};
\node[diamondnode]      (is3)   at (-3,1.25) {IS};

\draw[->][thick] (-1.6,0.25) -- (is1.east);
\draw[->][thick] (-1.6,0.25) -- (is2.east);
\draw[->][thick] (-1.6,0.25) -- (is3.east);

\draw[->][thick] (fs1.west) -- (5.5,0.25);
\draw[->][thick] (8.5,0.25) -- (fs1.east);
\draw[->][thick] (7.5,0.25) -- (7.5,1);
\draw[->][thick] (7.5,1) -- (6,1);
\draw[->][thick] (6,1) -- (6,0.25);

%explain

\filldraw[color= yellow!50, draw=yellow, very thick] (0,-2.5) rectangle (1,-2);
\node[right] at (1,-2.25)  {Electric Vehicle};

\filldraw[red,thick] (4,-2.25) circle (4pt);
\node[right] at (4.25,-2.25)  {Passenger};

\end{tikzpicture}}
	\caption{\label{fig:SS} Single Server (SS) queue: electric vehicles queue up, waiting for incoming passengers. Passenger will always pick the first vehicle in line if there are some vehicles waiting. }
\end{figure}   

\subsubsection{Single Server (SS) queues - Departure}
At each departure point, vehicles queue up to wait for the arrival of the next passenger. If there is no vehicle at the departure point, then any passenger who arrives would leave immediately. We view node $i\in S$ as a First-Come-First-Serve (FCFS) single server (SS) queue, whose service rate is the passenger's arrival rate at station $i$: $\alpha_i$, i.e., the service time is exponentially distributed with mean $\frac{1}{\alpha_i}$, which is the same as the inter-arrival time of passengers at this point. As shown in Figure \ref{fig:SS}, if there are at least two electric vehicles at the departure node $i\in S$, one vehicle is being served in the server of the queue and others are waiting in the queue.
Once there is an arrival of a new passenger,  the vehicle in the server will finish its service and leave this node carrying one passenger, i.e., the arriving passenger will pick up the first vehicle in the line once arrival. When the server is idle, the second vehicle will enter the server and starts its service. This is similar to an airport taxi service when there is only one line. The first taxi is always on the server and it departs once a passenger arrives, the second taxi becomes the first one after the departure and the last-come taxi will stay at the end of the queue, which forms the FCFS discipline. 

\subsubsection{Infinite Server (IS) queues - Travel}
After departing from node $i \in S$, each vehicle will go towards its destination $j \in F$ selected by passenger with probability $p_{ij}$, where $\sum_{j} p_{ij}=1$, $p_{ii}=0$. We use an infinite-server (IS) queue connecting the origin and destination to model the travel time of passengers. We assume that the travel time is independent across all the passengers and follows a general distribution with mean $T_{ij}$, which is associated with the distance between departure point $i $ and charging point $j$. 

\subsubsection{Direct path - No Charging}
After the passenger is dropped off at the destination, the electric vehicle may need to be charged. Accordingly, we assume that with probability $\bar{p_j}$, the vehicle decides to be charged at node $j\in F$, and with probability $1-\bar{p_j}$, it goes directly to the following single server queue without waiting or charging.

\subsubsection{Finite server (FS) queues - Charging}\label{sub:fscharge}
If the vehicle decides to be charged, we use an FCFS finite server (FS) queue to model the charging process at charging point $j \in F$. As both chargers and spaces are limited at each station, we assume that the maximum number of vehicles charged simultaneously is $v_j$ at charging point $j$. All vehicles that decide to be charged at that node forms an FCFS queue to wait for charging. To simplify the analysis, we assume that the charging time follows an exponential distribution with mean $t_j$ for $j\in F$. After the charging process is over, charged electric vehicles will enter the following single server queue to wait for the passengers.

\subsection{Closed Queueing Network}
In this network, we have considered three types of nodes: Single server queue ($S$), Infinite server queue ($I$) and Finite server queue ($F$). Let $\mathcal{N} = S \cup I \cup F$ denote the set of all nodes, and denote $N=|\mathcal{N}|$. For each node $i\in \mathcal{N}$, let Parent$(i)$ be the direct origin of node $i$, i.e., as shown in the Figure \ref{fig:System}, Parent$(7)$ is node $6$ and  Parent$(2)$ is node $1$.

The routing matrix of vehicles between nodes can be written as follows: 
\begin{align}{\label{equ:r_ij}}
	r_{ij}= \left\{
	\begin{aligned}
		&p_{il},  \text{\space} i\in S, j\in I, l \in F, \\
		& \qquad \quad i=\mbox{Parent}(j), j= \mbox{Parent}(l)\\
		&\bar{p_j}, \text{\space} i\in I, j \in F, i=\mbox{Parent}(j) \\
		&1-\bar{p_k}, \text{\space} i\in I, k \in F, j \in S, \\
		& \qquad \quad i=\mbox{Parent}(k), k=\mbox{Parent}(j) \\
		&1, \text{\space} i\in F, j \in S, i=\mbox{Parent}(j) \\
		&0, \text{\space} \mbox{otherwise}
	\end{aligned}
	\right.
\end{align}
where the first case means that,  after selecting the destination $l \in F$ with probability $p_{il}$, passengers enter the associated roads and begin travelling at node $j\in I$ following departure from node $i \in S$. The second case indicates that vehicles will choose to charge at charging point $j \in F$ with probability $\bar{p_j}$ after exiting the roads $i \in I$. The third case denotes that vehicles will move directly to the departure point $j \in S$ after exiting the roads $i \in I$ with probability $1-\bar{p_k}$, which skips charging at $k\in F$. The fourth case indicates that all vehicles will move to the departure points $S$ if they finish the charging process within the same station.

When there are  $n_i \in \{0,\ldots,M\}$  vehicles at node $i \in \mathcal{N}$, the service rate at each node (the average number of vehicles finishing service and leaving this node per unit time) is as follows:
% \begin{align}{\label{equ:u_i}}
% 	u_{i}(n_i)= \left\{
% 	\begin{aligned}
% 		&\alpha_i,  n_i \geq 1, \text{\space} i \in S\\
% 		&\frac{n_i}{T_{jl}}, \text{\space} j\in S, i \in I, l\in F\\
% 		& \quad \quad j=\mbox{Parent}(i), i=\mbox{Parent}(l)\\
% 		%&\frac{n_i}{t_i}, \text{\space} i\in F, n_i < v_i\\
% 		&\frac{\min\{n_i,v_i\}}{t_i}, \text{\space} i\in F
% 	\end{aligned}
% 	\right.
% \end{align}
\begin{align}{\label{equ:u_i}}
	u_{i}(n_i)= 
	\begin{cases}
		\alpha_i &  n_i \geq 1, \text{\space} i \in S\\
		0 &  n_i =0, \text{\space} i \in S\\
		\frac{n_i}{T_{jl}} & \text{\space} j\in S, i \in I, l\in F, j=\mbox{Parent}(i), i=\mbox{Parent}(l)\\
		\frac{\min\{n_i,v_i\}}{t_i} & \text{\space} i\in F
	\end{cases}
\end{align}
where the first case means that, if there is at least one vehicles at station, the departure process only depends on the the arrival of passengers: $\alpha_i$ and does not depend on the number of vehicles: $n_i$ at this node.  The second case indicates that all travel times are independent from each other and $T_{jl}$ is the mean travel time from departure point $j\in S$ and charging point $l\in F$. The third case means that if the number of vehicles willing to be charged $n_i$ is larger than the number of chargers $v_i$ at charging point $i \in F$, they need to form a FCFS queue to wait until a vehicle finishes charging and the charger becomes available, while $t_j$ is the average charging time at charging point $j\in F$.

\section{Closed Queueing Network Analysis}\label{sec:queue}

In this section, we introduce some results in a closed queueing network and in particular, the BCMP network. In our model from the last section, SS queues and FS queues fall into the type-I queues and IS queues belong to the type-III queue in BCMP network \cite{gelenbe1998introduction}. Therefore, this model falls into the class of closed BCMP network, which has the product-form solution to the stationary state distribution, because these queues are quasi-reversible \cite{balsamo2000product}.

Given the fleet size $M$, i.e. there are $M$ electric vehicles routing within the network  and no electric vehicle enters or leaves the system, the associated continuous-time Markov process has the following space 
\begin{align*}
	\mathcal{S}=\Bigg\{(n_1,n_2,\ldots,n_N):\sum_{i=1}^{N} n_i=M, n_i \in \Na\Bigg\},
\end{align*}
where $n_i$ is the number of vehicles at node $i\in \mathcal{N}$. Since  the transition from one node to another is instant in our model, every vehicle must be at one node $i \in \mathcal{N}$.

Let $\lambda=(\lambda_1,\ldots,\lambda_N)$ denote the relative throughput at node $i \in \mathcal{N}$, which is defined as the relative average number of vehicles passing through the node per unit time. Since there are a fixed number $M$ of vehicles routing among the nodes, we have the following constraint (global balance equations):
\begin{align}{\label{equ:lambda}}
\lambda_i \sum_{k \in \mathcal{N}} r_{ik}=\sum_{j\in \mathcal{N}} \lambda_j r_{ji}, \text{\space} \forall i\in \mathcal{N},
\end{align}
where the probability of vehicle routing from node $i$ to node $j$ is $r_{ij}$ in \eqref{equ:r_ij}. With another constraint $\sum_{i \in \mathcal{N}} \lambda_i =1$, we can find the unique solution to  \eqref{equ:lambda} with respect to $\lambda$, which is also called visit ratio \cite[p.53]{balsamo2007queueing}.

%Note that $\lambda$ is not unique and has its value relatively ($c \lambda$ is also the solution to equation above where $c$ is a constant). We define $\sum_{i \in \mathcal{N}} \lambda_i =1$ to simplify the expression and now $\lambda$ is unique, which  is also called visit ratio.

It now follows from \cite{baskett1975open} that the stationary probability distribution of the resulting continuous time Markov process $P(n_1,n_2,\ldots,n_N)$ has the following product form :
\begin{align}{\label{equ:p_n}}
	&P(n_1,\ldots,n_N)=\frac{1}{G(M)}\prod_{i=1}^N \frac{\lambda_i^{n_i}}{\prod_{k=1}^{n_i} u_i(k)},
\end{align}
where $G(M)$ is the normalization constant in order to make its summation equal to one. A computational method to compute $G(M)$ is discussed in Subsection \ref{sub:convolution}.

From operational perspective, throughput and availability are the key performance indicator for the overall system. Throughput captures how many passengers are served per unit time. Availability is defined as the probability that at least one vehicle is available at the departure point. Due to the product form of the invariant distribution, these two quantities can be computed in closed form, and the expressions are given in the following lemma.

\begin{lemma}{\label{lemma:BCMP}}
	In a closed BCMP queueing network with $M$ vehicles and $N$ nodes, the throughput and availability are as follows
	\begin{enumerate}
		\item  The throughput of each node  $i\in \mathcal{N}$ (the average number of vehicles passing through node $i$ per unit time) is 
		\begin{align}{\label{equ:Lambda2}}
			\Lambda_i(M)=\lambda_i \Lambda(M),
		\end{align}
		where the system throughput of the network is 
		\begin{align}{\label{equ:Lambda1}}
			\Lambda(M)=\frac{G(M-1)}{G(M)}.
		\end{align}
		
		\item  The availability at departure points $S$, i.e., the probability that node $i \in S$ has at least one vehicle is
		\begin{align}{\label{equ:avail}}
			A_i(M) = P\{n_i\geq 1\} = 1-P\{n_i = 0\} =\frac{\lambda_i}{\alpha_i}  \Lambda(M).
		\end{align}
		
	\end{enumerate}
\end{lemma}
\begin{proof}{}
	We refer the reader to Appendix A.
\end{proof}

% We now discuss the passenger loss at each station. As defined above, $A_i(M)$ denotes the availability of vehicles at departure point $i \in S$. If there is at least one vehicle at departure point $i$, we claim that the status is available. Therefore, the availability is defined as the stationary state probability that node $i$ has at least one vehicle \cite{george2011fleet}. 

In \eqref{equ:avail}, with a high vehicle arrival rate $\lambda_i$ (high supply) and a low passenger arrival rate $\alpha_i$ (low demand) at departure point $i$, the service availability $A_i(M)$ will be relatively high. As defined before, we assume that there is a 'passenger loss' if there is no electric vehicle at departure point when a passenger arrives, i.e., the passenger leaves this system and try other modes of transportation. We show below that the probability of loss of a passenger at any time is related to the notion of availability introduced above.

\begin{lemma}{\label{lemma:loss}}
	If the fleet size is $M$, then the probability that there is a ``passenger loss" at departure point $i\in S$ is $1-A_i(M)$, where $A_i(M)$ is defined in equation \eqref{equ:avail}. 
\end{lemma}

\begin{proof}{}
	From Poisson arrival see time average (PASTA) \cite{wolff1982poisson}, the probability of the state as seen by an outside random observer is equal to the probability of the state seen by an arriving passenger under Poisson arrival. Since $1-A_i(M)$ is the probability that an outside random observer see that there is no vehicle at departure point $i \in S$, it is also the probability that an arriving passenger following the Poisson process finds no vehicle at departure point $i$, which leads to a passenger loss.
\end{proof}{}

In the following, we discuss the insensitivity property in the product form networks \cite[Section 3.4]{balsamo2000product}, which shows that the stationary distribution of the network $P(n_1,n_2,...,n_N)$ does not depend on the variance of service time distribution at infinite server nodes (IS), i.e., if the variance of travel time distribution between stations is changed, the average  performance metrics (throughput,  average waiting time,  average queue length) remains the same.

\begin{lemma}{\label{lemma:insen}}
    Consider a closed BCMP network introduced above. Let $\xi = (\xi_i)_{i\in I}$ denote the travel time distribution of the vehicles in the infinite server node $i\in I$, and let $\bar t_{\xi_i}$ be the mean of the distribution $\xi_i$. Let $P_\xi(n_1,n_2,\ldots,n_N)$ denote the corresponding stationary distribution. If $\xi'$ is another travel time distribution such that $\bar t_{\xi_i'} = \bar t_{\xi_i}$, then $P_{\xi} = P_{\xi'}$.
\end{lemma}

\begin{proof}{}
    The result follows from \cite[Corollary 4.1 \& Theorem 6]{chandy1977product} via station balance. This property holds because units receive service immediately upon entering the queue and their wait times are zero.
\end{proof}

A consequence of the above result is that the stationary distribution is dependent only on the mean of the service time distribution for infinite server queues; the precise distribution does not matter. As the average performance metrics (throughput, average waiting time, average queue length) can be derived from stationary distribution $P$, they are also independent of service time distribution for infinite server queues.

Furthermore, we can extend the insensitivity property to finite server queues (FS) if the number of vehicles $M$ is less than or equal to the number of chargers at finite server queues $v_i$. Intuitively, when the condition $M \leq v_i$ holds, the queue $i$ behaves as same as an infinite server queue, therefore the stationary state distribution only depends on the mean service time at finite server queue $i$.

\begin{lemma}{\label{lemma:finite_insen}}
	In closed BCMP network, if $M \leq v_i $, i.e., the number of vehicles in the network is equal to or less than the number of charger at node $i \in F$, then the  stationary state distribution $P(n_1,n_2,\ldots,n_N)$ depends  on  the  service  time  distribution at finite server queue $i$  only  through  its mean.
\end{lemma}

\begin{proof}{}
	This lemma follows from \cite[p.250 Condition 3]{baskett1975open}  and its Section 4.1, where only the mean service times appear in $P(n_1,n_2,\ldots,n_N)$ and any service time distribution with the same mean yields the same results as exponential service time distribution.
\end{proof}

\section{Optimal Fleet Sizing}\label{sec:fleet}

From the view of a  mobility-on-demand service provider, one critical variable is the size of an electric fleet, before launching service in one city. In this section, we want to develop a profit maximization problem with operating cost, by controlling the fleet size, while maintaining a certain quality of service. The problem in \cite{george2011fleet} is extended to more general cases with any travel time distribution, a finite number of chargers, convex operating cost function and location-specific availability requirements.   

As service providers can only make money when vehicles are traveling, we model its total revenue per unit time as $\sum_{i\in I} \Lambda_i(M) z_i $, where $z_i$ is the revenue per-service (one-way charge) when vehicles are in IS nodes $i \in I$.  $\Lambda_i(M)$ is the throughput of node $i$, i.e., the average number of service finished at node $i$ per unit time.  Besides, we define the operating cost (salary, maintenance, etc) per unit time as a convex (including linear) increasing function $g(M)$ with respect to the fleet size $M$. 

As there are various requirements of availability at different places (e.g., high availability at airports and downtown), we define $\epsilon=(\epsilon_1, \ldots, \epsilon_s)$ as the quality of service requirement in the system. At each departure point $i \in S$, the availability $A_i(M)$ defined in  \eqref{equ:avail} is greater than or equal to $1-\epsilon_i$. 

From a steady-state view of the system, we want to maximize the profit by controlling the fleet size $M$, while maintaining a certain quality of service $A_i(M)$. The optimization problem can be formulated as follows:
\begin{align}
	&\max_{M \in \Na} f(M) = \sum_{i\in I} z_i\Lambda_i(M) -g(M) \\
	&s.t. \text{\space}  A_i(M) \geq 1-\epsilon_i, \text{\space} \forall i \in S {\label{equ:constraint}}
\end{align}
In the next two lemmas, we show that the objective function $f$ is concave and that the above optimization problem is feasible. For a function $g:\Na\rightarrow\Re$ defined over the space of natural numbers, it is said to be concave \cite{shanthikumar1988second} if 
\begin{align*}
    f(M)+f(M+2) \leq 2 f(M+1) \text{ for all } M \in \Na.
\end{align*}

\begin{lemma}{\label{lemma:concave}}
The objective function $f:\Na\rightarrow\mathbb{R}$ is concave in $M$. 
\end{lemma}

\begin{proof}
 For the system under exponential travel time distribution without the charging stations, this result is established in \cite[Theorem 2, p. 202]{george2011fleet}. We show that essentially the same argument holds for our case with the charging stations and any travel time distribution in \ref{app:concave}.
\end{proof}

In the following lemma, we show that, if there are more vehicles in the system, the availability at every departure point will increase. 

\begin{lemma}{\label{lemma:avail}}
	The availability function $A_i(M)$ at each departure point $i\in S$ is non-decreasing with $M$.
\end{lemma}

\begin{proof}{}
	From the first part of proof in Lemma \ref{lemma:concave} and \eqref{equ:avail} $A_i(M) =\frac{\lambda_i}{\alpha_i}  \Lambda(M)$, we find that $A_i(M)$ is non-decreasing with $M$.
\end{proof}

Therefore, if there exists $M_{\epsilon_i}$ such that $ A_i(M) \geq 1-\epsilon_i$ holds for $ M \geq M_{\epsilon_i}$, let $M_\epsilon = \max_{i \in S} M_{\epsilon_i}$, we can conclude that the constraint \eqref{equ:constraint} is satisfied for $ \forall M \geq M_\epsilon$.

\begin{theorem}{\label{thm:fleet}}
	If $g(M) \rightarrow \infty$ as $M \rightarrow \infty$, the optimization problem above has at most two solutions.
\end{theorem}
\begin{proof}
	Take the backward discrete derivative as
	\begin{align*}
		\Delta f(M)=f(M)-f(M-1)
	\end{align*} 
	We know from \eqref{equ:avail} that $\Lambda(M)$ is upper bounded. This yields 
	\begin{align*}
		f(M)=\Lambda(M) \sum_{i\in I} z_i \lambda_i-g(M)
	\end{align*}
	we have $f(+\infty) \rightarrow -\infty$. Since $f$ is concave from Lemma \ref{lemma:concave}, $f(M)$ is decreasing when $M$ is sufficiently large, then there exists a critical point, such that either (i) $\Delta f(M_1)=0$ or (ii) $\Delta f(M_2)>0 $ and $\Delta f(M_2+1)<0$. Then $M^*=M_1$ or $(M_1-1)$ in the first case or $M^*=M_2$ in the second case. 
\end{proof}

From the solution provided in the proof above, we find that the optimal fleet size $M^*$ is determined by various parameters: routing probability matrix $r_{ij}$, service rate $u_i(n_i)$, revenue per service $z_i$, fleet operating cost $g(M)$ and quality of service requirement $\epsilon$. In Section \ref{sec:numeric}, we show through numerical simulation  the effect of these parameters on the optimal fleet size. We summarize the procedure  to find the optimal fleet size in Algorithm \ref{alg:fleet}.

\begin{algorithm}
	\SetAlgoLined
	\SetKwInOut{Input}{Input}\SetKwInOut{Output}{Output}
	\Input{$f(M)$, $A_i(M)$, and $\epsilon_i$ for all $i \in S$}
	\Output{$M^*$}
	\BlankLine
	\For{$M \xleftarrow{} 1$ \KwTo $\infty$}{
	    \If{ $A_i(M) \geq \epsilon_i$, $\forall i \in S$}{
    	    \If{$\Delta f(M) = 0$}{
    	        \eIf{$A_i(M-1) \geq \epsilon_i$, $\forall i \in S$}{
    	            \Return{$M$ and $(M-1)$}\;
    	            }{
    	            \Return{$M$}\;}
    	    }
    	    \If{$\Delta f(M) > 0$ {\bf and} $\Delta f(M+1) < 0$ }{
    	        \Return{$M$}\;
    	    }
	    }
	}
% 	\Return{$\emptyset$}\;
	\caption{Optimal Fleet Sizing}{\label{alg:fleet}}
\end{algorithm}

% the optimal fleet size $M^*$ can be summarized as follows:
	
% 	If $M_1$ exists, i.e. $\Delta f(M_1)=0$ holds, then
% 	\begin{align*}
% 		M^*=
% 		\begin{cases}
% 			M_\epsilon      & \text{if } M_\epsilon \geq M_1\\
% 			M_1 \text{ and }  (M_1-1)  &  \text{if } M_1-1 \geq M_\epsilon\\
% 			\emptyset       & \text{if } M_{\epsilon_i} \text{\space does't exist for some } i \in S
% 		\end{cases}
% 	\end{align*}
% 	Otherwise, $M_2$ exists, i.e. $\Delta f(M_2)>0 $ and $\Delta f(M_2+1)<0$ hold, then
% 	\begin{align*}
% 		M^*= 
% 		\begin{cases}
% 			M_\epsilon &  \text{\space if } M_\epsilon \geq M_2\\
% 			M_2  & \text{\space if } M_2-1 \geq M_\epsilon\\
% 			\emptyset & \text{\space if } M_{\epsilon_i} \text{\space does't exist for some } i \in S
% 		\end{cases}
% 	\end{align*}
	
% 	The case of no solution is because $A_i(M)$ may not reach $1-\epsilon_i$ even when $M\rightarrow \infty$, where $\epsilon_i$ is extremely small or $\frac{\lambda_i}{\alpha_i}$ is relatively small for some $i$. 
	
% Another interesting point drawn from this section is that, if the service rate at one node $i$ is changed, e.g., add one more charger at one charging point $i \in F$, then the throughput of all nodes will change as well as the system profit. Therefore, does there exist an optimal charger allocation scheme in order to maximize the profit? In the following section, we try to find the answer.   

\section{Charger Allocation}\label{sec:charge}

In practice, there are usually very limited spaces for charging in the downtown area (the rent is high) and building charging infrastructure takes a non-trivial amount of money and time. As a result, the service provider needs to decide on the location of the charging stations and the number of chargers to be installed at each charging station. Intuitively, if more chargers are built, electric vehicles will spend less time waiting or driving around looking for unoccupied chargers, which leads to more availability. On the other hand, building and operating more chargers will increase operating costs. Therefore, there needs to be a trade-off between quality of service and operating cost.

We model it as a profit maximization problem, by controlling $V$, where $V=(v_1,v_2,\ldots,v_f)^T$ is the vector of the number of chargers at each charging point $i\in F$. Throughout this section, we fix the fleet size to $M$ and consider the throughput and availability as a function of $V$. Towards this end, by a slight abuse of notation, we let the throughput at node $i$ be denoted by $\Lambda_i(V)$, the system throughput be denoted by $\Lambda(V)$, and the availability by $A_i(V)$.

Let $\hat{V}=(\hat{v}_1,\ldots,\hat{v}_f)$ be the maximum number of chargers allowed at each point due to limited space or power constraint. We further assume that all chargers are identical and have the same charging speed in this section, i.e., for mean charging time defined in \eqref{equ:u_i}, we have $t_i = t_j$ $\forall i, j \in F$.

Let $z_i$ be the average revenue per service at node $i \in I$. We further assume that there is a penalty of $\beta_k$ dollars if there is a passenger loss, i.e., passenger finds no vehicle at departure point $k \in S$ and leaves the system. From Lemma \ref{lemma:loss}, the penalty per unit time at node $k$ is $\beta_k \alpha_k (1-A_k(V))$, where $\alpha_k$ is the passenger arrival rate. Let $c_i$ be the average cost for maintaining one charger at charging node $i\in F$ per unit time, which captures different rent and electricity rates at various places. Thus, the operating cost of chargers is $c_j v_j$ at charging point $j \in F$. The resulting optimization problem can be formulated as follows:

\begin{align}
	&\max_{V \in \Na^F} \sum_{i\in I} \Lambda_i(V) z_i -\sum_{k\in S}\beta_k \alpha_k (1-A_k(V)) -\sum_{j \in F}c_j v_j  \\
	& s.t. \text{\space} V \leq \hat{V}  
\end{align}
where the objective function is the revenue minus penalty due to loss of a passenger and the operating cost. We want to maximize it by controlling the number of chargers $V$ at various charging points.  The constraint means that the number of chargers at each charging point $i \in F$ is upper bounded by $\hat{v_i}$. 

We now simplify the objective function. Similar to \eqref{equ:Lambda1} in Lemma \ref{lemma:BCMP}, we define the system throughput under a fixed $M$ as 
\begin{align}
	\Lambda(V) = \frac{G(M-1,V)}{G(M,V)},
\end{align}
where $G$ is the normalization constant introduced in \eqref{equ:p_n}. Therefore, following \eqref{equ:Lambda2}, the actual throughput of each node $i \in I$ is 
\begin{align*}
	\Lambda_i(V)=\lambda_i \Lambda(V).
\end{align*}
Using \eqref{equ:avail}, the above optimization problem can be rewritten as
\begin{align}
	&\max_{V \in \Na^F} \quad h(V) := \Lambda(V) \bar{Z} -\sum_{j \in F}c_j v_j -\sum_{k \in S}\beta_k \alpha_k \\
	&s.t. \text{\space} V \leq \hat{V} {\label{equ:upperbound}}
\end{align}
where $\bar{Z}$ is independent from $V$ and defined as follows,
\begin{align}
	\bar{Z}=\sum_{i \in I}\lambda_i z_i+\sum_{k \in S}\lambda_k \beta_k
\end{align}{}
Let $e_j\in\{0,1\}^f$ denote the unit vector with $1$ along the $j^{th}$ dimension and $0$ otherwise. In the following theorem, we show the concavity of objective function with respect to $v_j$, the number of chargers at charging point $j$, for all $ j \in F$.  
\begin{theorem}{\label{thm:concavity}}
The following holds:
\begin{enumerate}
    \item The map $v_j\mapsto \Lambda(V)$ is an increasing concave function for all $j \in F$.
    \item The objective function $h(V)$ satisfies $h(V)+h(V+2 e_j) \leq 2 h(V+e_j)$ for all $j\in F$.
    \item The map $v_j\mapsto A_i(V)$ is an increasing concave function for all $i, j \in F$.
\end{enumerate}
\end{theorem}
\begin{proof}{}
	See \ref{app:concavity}.
\end{proof}

This third part of the theorem points towards an interesting property of the network: adding chargers at any charging point will increase the system throughput and the availability of any departure point in the system. Therefore, service providers can firstly allocate chargers to the charging points which can bring high system throughput increment at a low cost. 

We now outline an algorithm, proposed in Section 4 of \cite{shanthikumar1988server}, that computes an approximately optimal charger allocation in Algorithm \ref{alg:charger}. This algorithm is inspired by the marginal allocation algorithm of \cite{fox1966discrete}. The underlying idea for this algorithm is to identify the location where adding one more charger leads to the maximum increment in the profit. This process is continued until either the increment becomes negative or the upper bound is reached.

% \begin{algorithm}
% 	\SetAlgoLined
% 	\SetKwInOut{Input}{Input}\SetKwInOut{Output}{Output}
% 	\Input{$h(V)$, $A(V)$, $\hat{V}$}
% 	\Output{$V^*$}
% 	\BlankLine
% 	$k=1$, $V^k=(1,1,\ldots,1)$, $\mathcal{F}= \{1,2,\ldots,f\}$\;
% 	\While{$V^k \leq \hat{V}$}{     
% 		\If{$v^k_j= \hat{v_j}$}{
% 			$\mathcal{F} \xleftarrow{} \mathcal{F} - \{j\}$\;
% 		}      
% 		\If{$max_{j \in \mathcal{F}} h(V^k+e_j)-h(V^k) \leq 0$}{
% 			\Return{$V^k$}\;
% 		}    
		
% 		$j^*=argmax_{j \in \mathcal{F}}  \text{\space} h(V^k+e_j)-h(V^k)$\;
% 		$V^{k+1} \xleftarrow{} V^{k}+e_{j^*}$\;
% 		$k \xleftarrow{} k+1$\;
% 	}
% 	\canwe{\Return{$\hat{V}$};}
% 	\caption{Charger Allocation Algorithm}{\label{alg:charger}}
% \end{algorithm}

\begin{algorithm}
	\SetAlgoLined
	\SetKwInOut{Input}{Input}\SetKwInOut{Output}{Output}
	\Input{$h(V)$, $A(V)$, $\hat{V}$}
	\Output{$V^*$}
	\BlankLine
	$k=1$, $V^k=(1,1,\ldots,1)$, $\mathcal{F}= \{1,2,\ldots,f\}$\;
	\While{$V^k \leq \hat{V}$}{     
		\If{$v^k_j= \hat{v_j}$}{
			$\mathcal{F} \xleftarrow{} \mathcal{F} - \{j\}$\;
		}
		$m = \max_{j \in \mathcal{F}} h(V^k+e_j)-h(V^k)$ \; 
		$j^*=\arg\max_{j \in \mathcal{F}}  \text{\space} h(V^k+e_j)-h(V^k)$\;
		\eIf{$m> 0$}{
		    $V^{k+1} \xleftarrow{} V^{k}+e_{j^*}$\;
		    $k \xleftarrow{} k+1$\;
		}{
		    \Return{$V^k$}\;
		}
	}
	\Return{$\hat{V}$};
	\caption{Charger Allocation Algorithm}{\label{alg:charger}}
\end{algorithm}

% \if 0
% \textbf{Charger Allocation Algorithm: }

% *************************algorithm format

% 1) Let $V^1=(1,1,\ldots,1)$, $k=1$, $\mathcal{F}= \{1,2,\ldots,f\}$

% 2) $V^{k+1}=V^{k}+e_{j^*}$, where $e_{j^*}$ is the $j^*$th unit vector and $j^*=argmax_{j \in \mathcal{F}}  \text{\space} h(V^k+e_j)-h(V^k)$

% 3) If $v^k_j= \hat{v_j}$, then  $\mathcal{F} \xleftarrow{} \mathcal{F} - \{j\}$.

% 4) If $A_i(V^k) \geq 1-\epsilon_i$ holds for all $i \in S$, and $max_{j \in \mathcal{F}}  \text{\space} h(V^k+e_j)-h(V^k) \leq 0$, then terminate the iteration. $V^*=V^k$.

% 5) If $V^k = \hat{V}$, then terminate the iteration. $V^*=V^k$.

% 6) Go back to step 2.

% \fi 
% \canwe{
% This algorithm addresses both high accuracy and efficiency in performance analysis. As it is a high dimension integer programming problem, this algorithm significantly reduces the computation complexity by leveraging the queueing model and concavity property.  Originally, the number of possible allocations is ${|V|-1\choose |F|-1}$. As a contrast, Algorithm \ref{alg:charger} only requires at most $|V|\times|F|$ profit evaluations. 
% }
If there are only two finite-server queues in the system whose number of chargers may change, we can guarantee that the solution found by the above algorithm is optimal. For the general case, its optimality remains a conjecture \cite{shanthikumar1988server}, which states that the solution found through this heuristic is usually the optimal solution, a proof of optimality is not available. 

\begin{theorem}\label{thm:v2}
If $|F|=2$,  then Algorithm \ref{alg:charger} generates the optimal solution.
\end{theorem}
\begin{proof}{}
Our proof follows the approach developed in \cite[Proposition 2, p. 339]{shanthikumar1988server}. We first show that our objective function $h$ is concave and supermodular. The proof of convergence of the algorithm proposed in \cite{shanthikumar1988server} is for a fixed number of servers. On the other hand, we relax this constraint, since in our case we can put as many chargers as possible. A detailed proof is presented in \ref{app:v2}.
\end{proof}

\section{Charger Selection}\label{sec:select}
For electric vehicles, two different types of chargers are available: fast chargers and slow chargers. The fast chargers can charge the vehicle rather quickly; it can charge a vehicle from 20\% charge to 80\% charge within 30 minutes. On the other hand, slow chargers would require over 6-8 hours to do the same. Intuitively, one would conjecture that having one fast charger is better than two slow chargers. In this section, we identify the conditions under which it is beneficial to have two slow chargers as opposed to one fast charger to improve the overall throughput of the system. The key insight we get here is that the throughput of the system is dependent on the wait time for the vehicles as well as the charging time of the vehicles. Having one fast charger certainly reduces the charging time, but it can potentially increase the wait time if the {\it coefficient of variation} of the distribution of charging time of the vehicles is somewhat larger than a threshold. It should be noted that the installation cost of a fast charging infrastructure is significantly higher than installing multiple slow chargers.

Suppose that the service provider has two possible options:
\begin{enumerate}
	\item Install one fast charger with mean charging time $t_0$.
	\item Install two slow chargers with mean charging time for each charger is $2t_0$.
\end{enumerate}{}

% If the charging time distribution is non-exponential, the stationary probability of the closed queueing network becomes non-product-form, and only approximation algorithms exist.

For ease of analysis, we discuss the charger selection problem by comparing individual queues under Poisson arrival. Let $\alpha_0$ denote the Poisson arrival rate of vehicles that need to be charged. The mean charging rate (number of vehicles charged per unit time) for option 1 is $\mu_1(n)=\frac{1}{t_0}$.  For option 2, the mean charging rate becomes
\begin{align}
	\mu_2(n)= \left\{
	\begin{aligned}
		&\frac{1}{2 t_0}, \text{\space} n=1\\
		&\frac{1}{t_0},  \text{\space} n \geq 2
	\end{aligned}
	\right.
	\label{equ:mu2}
\end{align}	
where $n$ is the number of vehicles at this charging point.

According to the model we assume in Subsection \ref{sub:fscharge}, we assume that upon arriving at a charging node, all the vehicles form an FCFS queue to wait for charging. Let $\gamma = \frac{\alpha_0}{\mu}$ denote the utilization of the queue, then both options have the same utilization, $\gamma_1=\gamma_2$.

Let $D_1$ and $D_2$ denote the average time delay of a vehicle at this node for both options (including waiting time and charging time). We show in the next section (see \eqref{equ:throughput} and the discussion that follows) that a smaller delay increases the system throughput of the closed queueing network. We want to find which option has a lower delay, thus it will have higher throughput, and as a consequence, a higher profit and better quality of service. 

Intuitively, one fast charger outperforms two slow chargers because $\mu_1(1) > \mu_2(1)$. In the following lemma, we show that it is true for exponential charging time, however, it may not hold for some general charging time distribution, which is proved in the following theorem. 

\begin{lemma}{}
	If charging time distribution is exponential, then $D_1 <D_2$, $\forall  \text{\space} \gamma \in (0,1)$,  
\end{lemma}

\begin{proof}{}
	For $M/M/1$ queue with arrival rate $\alpha_0$, service rate $\mu_1$ and utilization $\gamma_1=\frac{\alpha_0}{\mu_1}$, the average waiting time in queue is $w_1=\frac{\gamma_1}{\mu_1 (1-\gamma_1)}$ from \cite[p.82]{smith2018introduction}. Therefore, the average delay for option 1 is $D_1=w_1+t_0=\frac{t_0}{1-\gamma_1}$.
	
	For $M/M/2$ queue with  arrival rate $\alpha_0$, service rate $\mu_2(n)$ and utilization $\gamma_2=\frac{\alpha_0}{2\mu_2(1)}$, the average waiting time is queue is $w_2=\frac{\gamma_2^2}{\mu_2(1) (1-\gamma_2^2)}$, \cite[p.87]{smith2018introduction}. Thus the average delay for option 2 is $D_2=w_2+ 2t_0=\frac{t_0}{ 1-\gamma_2^2}$.
	
	As $\gamma_1=\gamma_2 \in (0,1)$, we can conclude that $D_1<D_2$, $\forall \gamma \in (0,1)$.
\end{proof}

Let $c^2=\frac{Variance}{Mean^2}$ denote the squared coefficient of variance of charging time, which measures the dispersion of the charging time distribution. 
\begin{theorem}\label{thm:d1d2}
	For any $\gamma \in (0,1)$, there exists a charging time distribution such that $D'_1 >D'_2$ for all  $c^2 >1+ \frac{2}{\gamma}$. 
\end{theorem}
\begin{proof}{}
	We prove this theorem by providing an extreme case of charging time distribution. 	For option 1, called $M/T_1/1$ queue,  the charging time follows the distribution
	\begin{align*}
		t_1= \left\{
		\begin{aligned}{}
			& \text{exp}(p_0/t_0 ) \text{\space} \text{w.p.} \text{\space} p_0\\
			& 0 \text{\space} \text{ w.p.} \text{\space} 1-p_0
		\end{aligned}
		\right.
	\end{align*}
	where exp($p_0/t_0$) denotes the exponential distribution with mean $t_0/p_0$. Then the mean of $t_1$ is $t_0$ and squared coefficient of variance is $c_1^2=\frac{2}{p_0}-1$.
	From P-K formula, we have $D'_1=\frac{t_0}{p_0}\frac{\gamma_1}{1-\gamma_1}+t_0$.
	
	For option 2, $M/T_2/2$ queue, its charging time for each charger follows the same distribution of option 1 with double mean, i.e.
	\begin{align*}
		t_2= \left\{
		\begin{aligned}{}
			& \text{exp}(p_0/2 t_0 ) \text{\space} \text{w.p.} \text{\space} p_0\\
			& 0 \text{\space} \text{ w.p.} \text{\space} 1-p_0
		\end{aligned}
		\right.
	\end{align*}
	So the mean of $t_2$ is $2t_0$ and squared coefficient of variance is $c_2^2=\frac{2}{p_0}-1$, same as $c_1^2$.
	
	In the following, we show that the average waiting time is $w'_2=w_2/p_0$, where $w_2$ is the average waiting time for $M/M/2$ queue. We notice that the waiting time experienced by non-zero jobs and zero-sized jobs is the same, as scheduling does not depend on size. We ignore the zero-sized jobs to calculate the waiting time of non-zero sized jobs. It is equivalent to the waiting time distribution of $M/M/2$ system with arrival rate $p_0 \alpha_0$ and mean charging time $\frac{2t_0}{p_0}$, which can be viewed as a $M/M/2$ system with arrival $\alpha_0$ and mean charging time $t_0$ slowed by factor $\frac{1}{p_0}$.
	Therefore, the average waiting time for $M/T_2/2$ system is $\frac{1}{p_0}$ times $M/M/2$ with arrival rate $\alpha_0$ and mean charging time $2t_0$, i.e., $w'_2=w_2/p_0=\frac{\gamma_2^2}{ p_0 u_2(1) (1-\gamma_2^2)}$. Finally, $D'_2=w_2'+2 t_0=\frac{2 t_0\gamma_2^2}{ p_0  (1-\gamma_2^2)}+2 t_0$.
	
	As $\gamma_1=\gamma_2$, we have $D'_1 > D'_2$ is equivalent to $\frac{\gamma}{1+\gamma} >p_0>0$. As $c_1^2=c_2^2=\frac{2}{p_0}-1$, we have $c^2>1+\frac{2}{\gamma}$.
\end{proof}

The above theorem indicates that two slow chargers can result in a lower overall delay than one fast charger, especially when the variance in charging time is relatively larger than the mean charging time. As stated previously, this happens because the large variance in charging time leads to a long waiting time, which reduces the waiting time by adding one more charger. Although the average charging time is doubled due to slow chargers, the decrease in waiting time is more significant than the increase of charging time, thus the total delay may decrease with two slow chargers.

Admittedly, the distribution constructed in the proof of Theorem \ref{thm:d1d2} is not representative of the actual charging time distribution. Nonetheless, we have found through simulations that for various distributions of charging time, there is a distribution dependent threshold for $c^2$, beyond which two slow chargers have lower delay (wait time plus charging time) than one fast charger. Indeed, in Subsection \ref{sub:onefast}, we show numerically the case for gamma distributed charging times and Inverse Gaussian distributed charging times.

\section{Computational Algorithms}\label{sec:algorithm}
In this section, we introduce some efficient algorithms for performance analysis, especially for large-scale networks. 

\subsection{Convolution Algorithm}\label{sub:convolution}

In order to compute the stationary state probability, the normalizing constant $G(M)$ is required as stated in \eqref{equ:p_n}. Explicitly, $G(M)$ has the following expression.
\begin{align}{\label{equ:gm}}
	G(M)=\sum_{n_1+\ldots+n_N=M} \prod_{i=1}^N \frac{\lambda_i^{n_i}}{\prod_{k=1}^{n_i} u_i(k)}
\end{align}

Direct computation of $G(M)$ as a summation over all possible states, which has a cardinality of ${N+M-1 \choose N-1}$, takes an exponential time to compute. However, we can use a convolution algorithm, which significantly reduces the complexity by developing an iterative algorithm.

Following the definition in Section \ref{sec:queue}, we  assume that there are $M$ vehicles and $N$ nodes in the system. For the case where $n_i$ vehicles at node $i \in N$, we define
\begin{align}
	k_i(n_i)=\frac{\lambda_i^{n_i}}{\prod_{k=1}^{n_i} u_i(k)}
\end{align}
Let $G_N(M)$ denote the normalizing constant of network with $M$ vehicles and $N$ nodes, then we have \cite{buzen1973computational}
\begin{align}
	G(M):=G_N(M)= k_1 * k_2 *\ldots * k_N (M)
\end{align}
where the convolution $k_1 * k_2 (m)$ of two functions $k_1$ and $k_2$ is defined by
\begin{align}
	k_1 * k_2 (m) = \sum_{i=0}^{m} k_1(i)k_2(m-i), \text{\space} m \geq 0
\end{align}
We can write the recursive relation in another way, for each $j=1,2,\ldots, N$, we have 
\begin{align}
	G_j(m)=k_j * G_{j-1} (m), \text{\space} 0 \leq m \leq M
\end{align}

Therefore, we can get the stationary probability distribution $P$ from \eqref{equ:p_n}, the throughput $\Lambda(M)$ from \eqref{equ:Lambda1}, and the marginal probability $p_i(n_i)$ for each node from \eqref{equ:marginal}.

\subsection{Mean Value Analysis}\label{sub:meanvalue}

As $M$ becomes larger, the convolution still takes a large amount of time, especially when the network has several load-dependent queues. If we are only interested in the average performance metrics,  we can utilize the Mean-Value analysis of closed queueing networks \cite{reiser1980mean} \cite{akyildiz1988mean} in order to compute the outcome easier. 

Let $D_i(M)$ denote the average system time (waiting time and charging time) of a passenger at node $i \in \mathcal{N}$. According to three different scheduling schemes (infinite, single, finite server), we have

\begin{align}	\label{equ:system}
	D_{i}(M)= \left\{
	\begin{aligned}
		&\frac{1}{u_i(1)}, \text{\space} i \in I\\
		&\frac{1+ L_i(M-1)}{u_i(1)},  \text{\space} i \in S\\
		&\frac{1+ L_i(M-1)+ s_i(M-1)}{v_i u_i(1)}, \text{\space} i\in F, 
	\end{aligned}
	\right.
\end{align}	
where $L_i(M-1)$ is the average number of vehicles (including the one in service) at node $i$, and $v_i$ is the number of chargers at finite-server nodes $F$.

We further define $s_i(M-1)$ as follows, which is the average number of idle chargers at node $i \in F$, then we have
\begin{align}\label{equ:idle}
	s_i(M-1)=\sum_{n_i=1}^{v_i-1} (v_i-n_i) p_i(n_i-1, M-1)%\\
	%=\sum_{n_i=0}^{v_i-2} (v_i-n_i-1) p_i(n_i, M-1)
\end{align}
where $p_i(n_i,M)$ is the marginal probability of $n_i$ vehicles at node $i \in F$ when the fleet size is $M$.

Therefore, we have the the system throughput as follows where $\lambda_i$ is defined in \eqref{equ:lambda}.
\begin{align}{\label{equ:throughput}}
	\Lambda(M)= \frac{M}{\sum_{i=1}^{N}\lambda_i D_i(M)}
\end{align}
Applying Little's Law, we can compute the queue length $L_i(M)$ by iteration without computing the normalization constant $G(M)$ as follows:
\begin{align}
	L_i(M)=\lambda_i \Lambda(M) D_i(M) %=***\frac{M \lambda_i D_i(M)}{\sum_{i=1}^{N}\lambda_i D_i(M)}
\end{align}

Therefore, we can iterate over $M$ to compute the throughput $\Lambda(M)$ from \eqref{equ:throughput}, where $D_i(M)$ only requires the information of last stage $L_i(M-1)$ and $s_i(M-1)$. There is one more expression $p_i(n_i,M)$ in \eqref{equ:idle} needing to specify.

In order to compute the marginal distribution $p_i(n_i, M)$ in \eqref{equ:idle}, for $i \in F$, we have to run another iteration with respect to $n_i$.

For $n_i=1,2,\ldots, M$, from local balance we have
\begin{align}
	p_i(n_i, M)=\frac{\lambda_i \Lambda(M)p_i(n_i-1,M-1)}{n_i u_i(1)}%=***\frac{\lambda_i M p_i(n_i-1,M-1)}{n_i u_i(1) \sum_{i=1}^{N} \lambda_i D_i(M)}
\end{align}
and
\begin{align}
	p_i(0,M)=1-\frac{1}{v_i}\Big(\frac{\lambda_i \Lambda(M)}{u_i(1)} 	+\sum_{n_i=1}^{v_i-1} (v_i-n_i)p_i(n_i,M)  \Big) 
	%=***1-\frac{1}{v_i}
	%\Big(\frac{M \lambda_i}{u_i(1) \sum_{i=1}^{N}\lambda_i D_i(M)}+\sum_{n_i=1}^{v_i-1} (v_i-n_i)p_i(n_i,M) \Big)
\end{align}

\if 0
The latter equation is from the equation of the average number of idle servers and both equations are used in the iteration of computing $D_i, i \in F$.
\fi

In this way, we can compute the throughput and availability faster than the convolution algorithm.

\subsection{General Passenger Inter-arrival Time Approximation}

In practice, passengers may not follow a Poisson arrival, that is,  the distribution of passenger inter-arrival time may not follow the exponential distribution. In this subsection, we show that we can approximate the passenger arrival process through a modification of the system equation in the Mean Value Analysis algorithm.

If the inter-arrival time is exponentially distributed, then it enjoys the memoryless property: We do not need to consider the remaining service time for vehicles at $i\in S$. If the inter-arrival time follows a general distribution with mean $\frac{1}{u_i(1)}$ and variance $\sigma^2$, we need to consider the remaining service time. Let $c_{\tau}^2=\sigma^2 u_i(1)^2$ denote the squared coefficient of variance of the inter-arrival time. Use $\gamma_i =\frac{\lambda_i}{u_i(1)}$ denote the relative utilization and the true utilization $\rho(M-1)=\gamma_i \Lambda(M-1)$ for the case of $M-1$ vehicles in the system. If the service time of single server nodes (SS) becomes general, then for each $i \in S$, the system time can be approximated by \cite[p. 291]{smith2018introduction}  and \cite[p. 253]{curry2010manufacturing} 
\begin{align}
	D_i(M)=\frac{1}{u_i(1)} &\Big(1+ L_i(M-1)-\rho(M-1)+
	\rho(M-1) \frac{1+c_{\tau}^2}{2} \Big), \text{\space} \forall i \in S
\end{align}

By replacing the second item in \eqref{equ:system} with the above equation, we can approximate the stationary distribution if the inter-arrival time has a general distribution. If $c_{\tau}^2=1$ holds, the above equations go back to the previous system time under exponential service time distribution in \eqref{equ:system}.

\section{Numerical Simulation}\label{sec:numeric}
    
    In this section, we run some large-scale simulations to validate our results. The first part develops a large-scale symmetric network and studies the asymptotic properties proved before. The second part focuses on the charger allocation on an asymmetric network capturing different characteristics of downtown and suburban areas. 
    The third part shows the effect of one fast charger and two slow chargers.
    
    \subsection{A Symmetric Network with 60 nodes}
    
Firstly, we consider one symmetric network with 60 stations with $p_{ij}= \frac{1}{59}$, i.e. after departure from one station, customers choose their destination equally between other stations. We further assume that one-third of EVs arriving at node $j$ decide to charge, while others choose to go to departure point directly without charging, i.e., $\bar{p_i}= \frac{1}{3}$, $\forall i \in F$. 
    
    In this network, we have 60 single server nodes (departure), 60 finite server nodes (charging) and 3540 infinite server nodes (travelling). From \eqref{equ:lambda}, we have the relative throughput as $\lambda_i=\frac{1}{420}$ for $ i \in F$, $\lambda_i=\frac{1}{140}$ for $ i \in S$ and $\lambda_i = \frac{1}{8260}$ for $i \in I$. The arrival rate of customers at each departure point $i\in S$ is $\alpha_i=10$ person per hour. The average time of travelling is  $T_{jl}=\frac{1}{3}$ hour per service (service rate: $3$ per hour), which follows a general distribution. The average time of charging is $t_i=0.5$ hour per charger $i \in F$ (service rate: $2$ per hour) and there are $v_i=2$ chargers at each charging station.
    
    \subsubsection{Fleet Sizing}
    
    We assume the average revenue per service is $z_i = \$30$ and the operating cost per vehicle is \$4 per hour $g(M)=4M$. Let $\epsilon_i= 20\%$ for $i \in S$, which means there is a minimal requirement of 80\% availability at each departure point. 
    
    Under above assumptions, the simulations at Figure \ref{fig:profit_m} and \ref{fig:avail_m} validate the concavity proved in Lemma \ref{lemma:concave} and Theorem \ref{thm:fleet} it also shows that the optimal fleet size is 763, with availability $87.2\%$.
    
\begin{figure}[!htb]
\centering
\scalebox{0.75}{
\includegraphics[width=\linewidth]{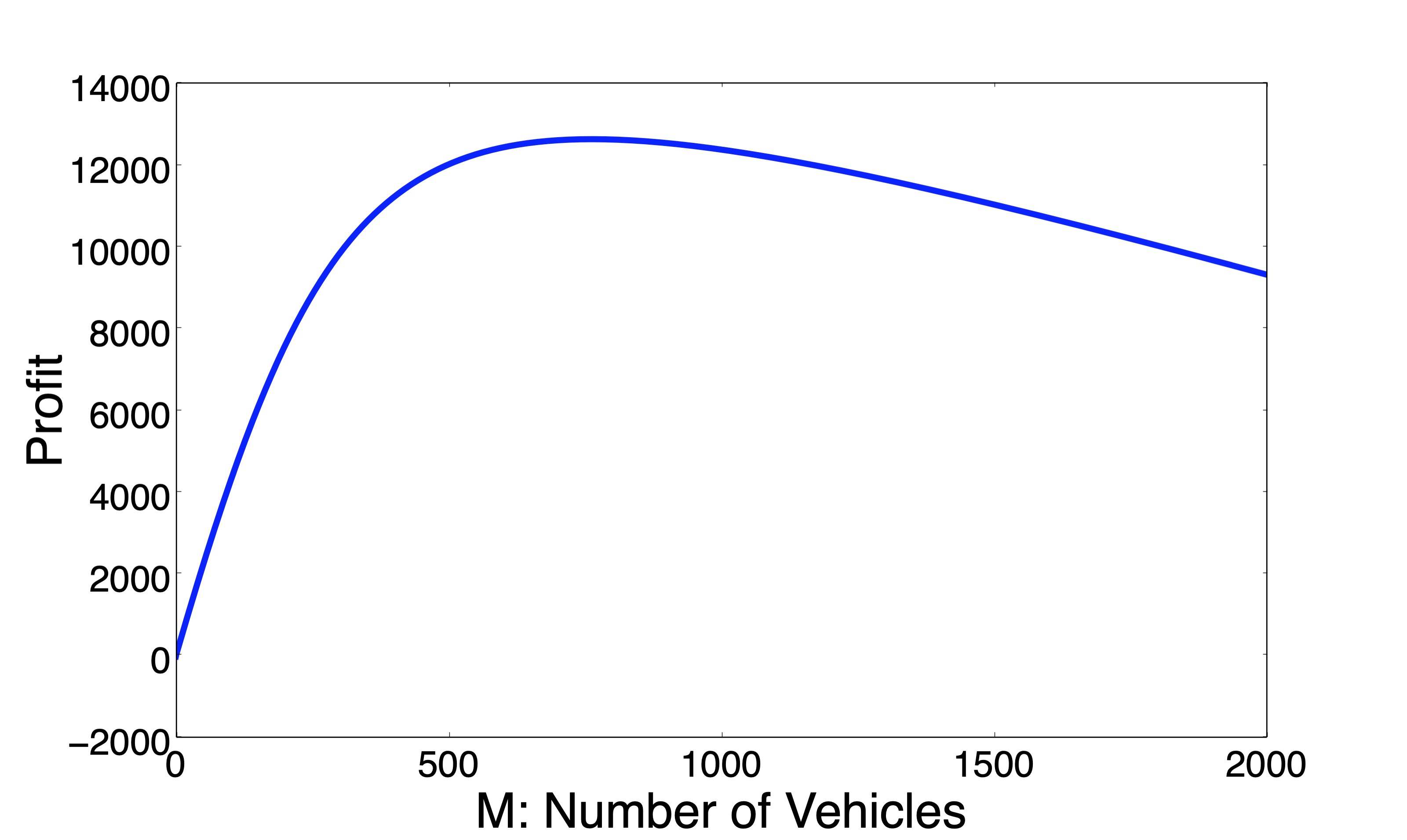}
}
\caption{Profit as a function of fleet size under 60 stations.}
\label{fig:profit_m}
\end{figure}

\begin{figure}[!htb]
\centering
\scalebox{0.75}{
\includegraphics[width=\linewidth]{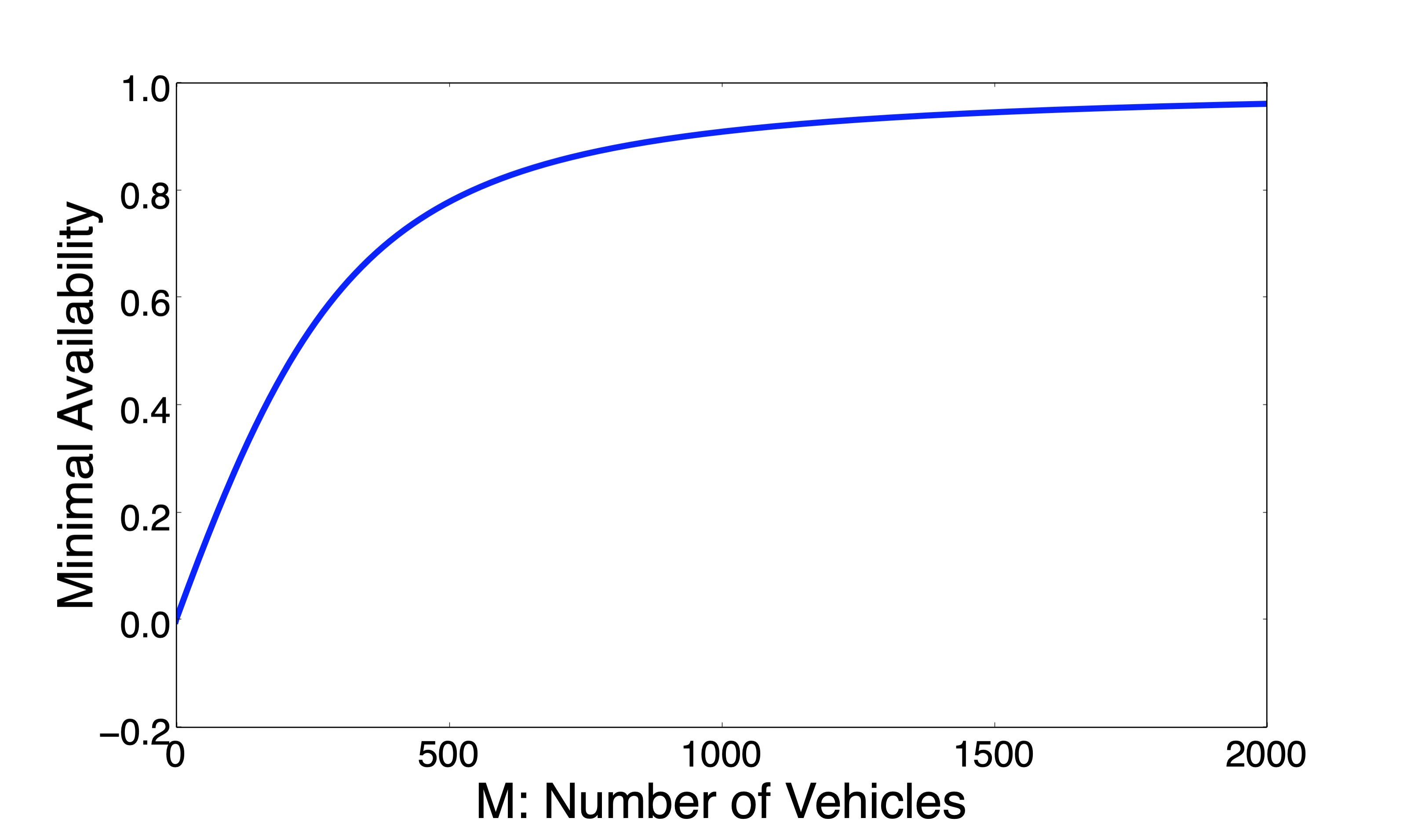}
}
\caption{Availability at departure points as a function of fleet size under 60 stations.}
\label{fig:avail_m}
\end{figure}

If the number of chargers is one at each station, the availability requirement is not satisfied (only $54.47\%$), thus we start from two chargers per station. As shown in Figure \ref{fig:opti_m_avail}, the optimal fleet size increases as the availability requirements increases. On the other hand, the optimal fleet size decreases as more chargers are provided at each charging point, as shown in Figure \ref{opti_m_num_charger}, and the curve remains stable when chargers are relatively high. This is because more chargers will decrease the waiting time before charging and provide more availability as less time is spent at charging points. 

% To be extreme, if there are infinitely many chargers and charging only takes very little time, we will only need fewer electric vehicles to serve the demand in the system. 
    
\begin{figure}[!htb]
\centering
\scalebox{0.75}{
\includegraphics[width=\linewidth]{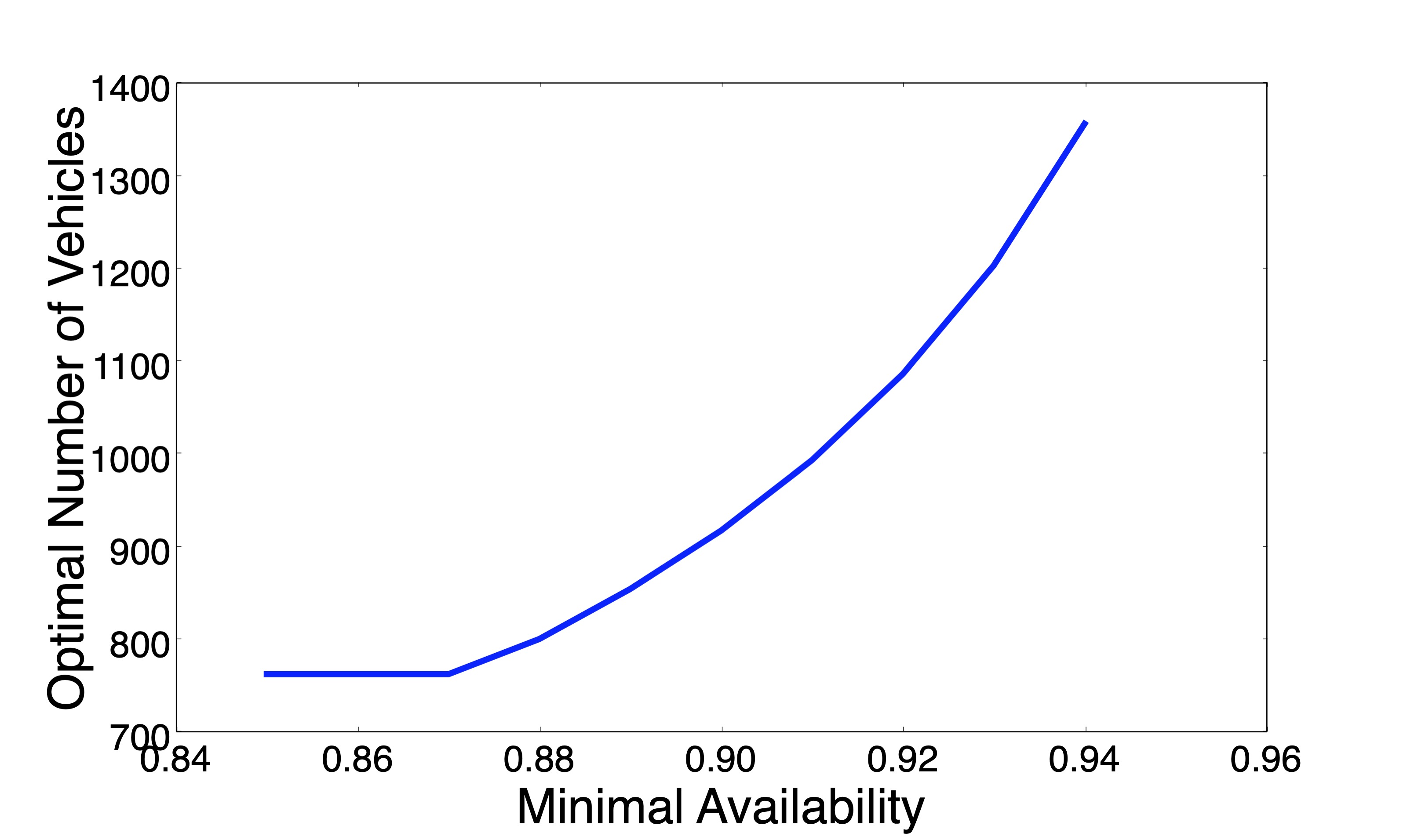}
}
\caption{Optimal fleet size increases dramatically with higher availability requirements.}
\label{fig:opti_m_avail}
\end{figure}

\begin{figure}[!htb]
\centering
\scalebox{0.75}{
\includegraphics[width=\linewidth]{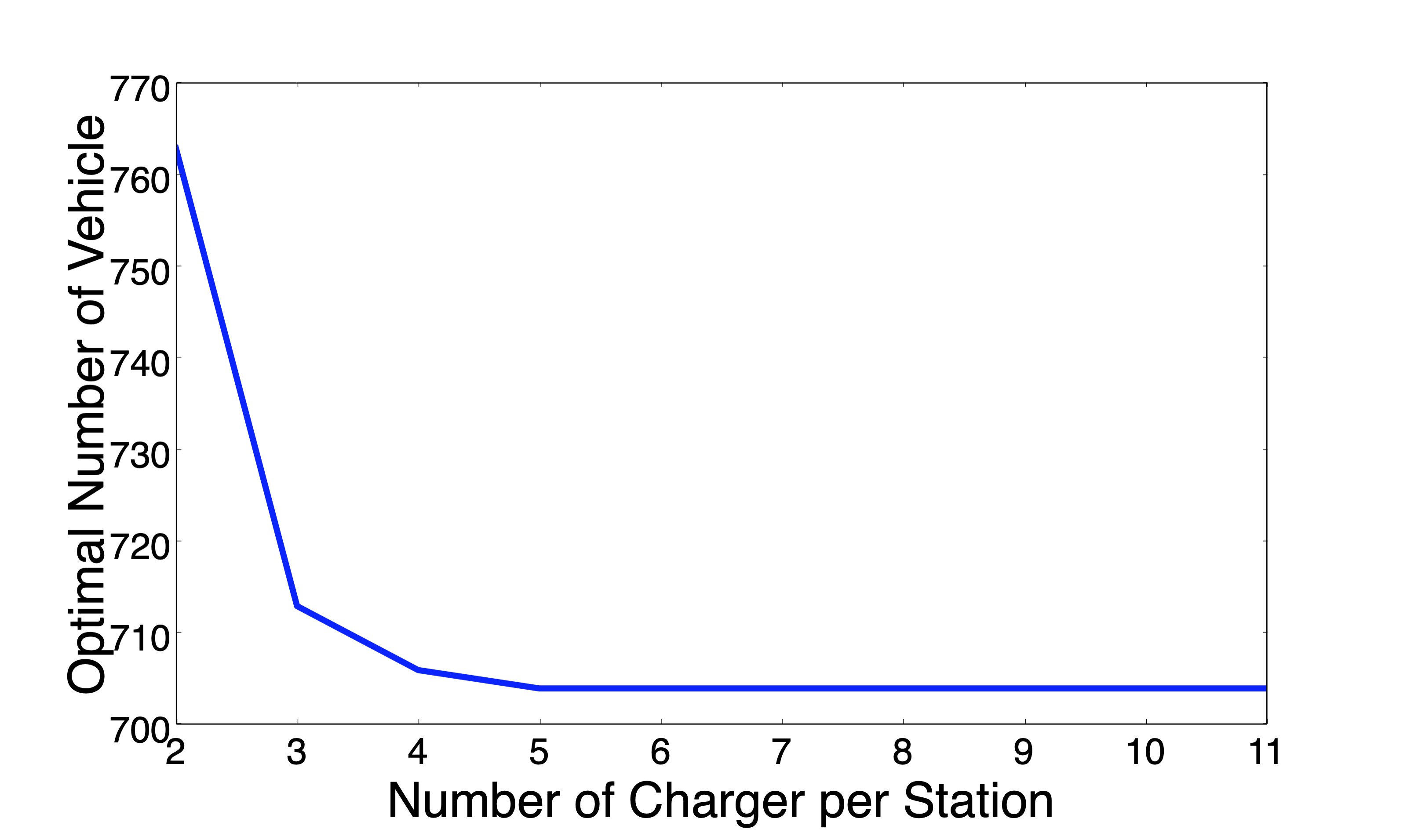}
}
\caption{Optimal fleet size decreases dramatically as number of chargers increases and it keeps flat when chargers are relatively large.}
\label{opti_m_num_charger}
\end{figure}

\subsubsection{Charger Allocation}
% As shown above, the number of chargers have a great influence on the system performances. In this subsection, we want to our theoretical findings proved in Section \ref{sec:charge}.
We now study the effect of number of chargers on the throughput, availability, and profit numerically for the example introduced above. We fix the fleet size to be $M=763$. Other parameters remain the same as in the last subsection such as transition probability, charging probability, arrival rate, charging time, travel time, and revenue per service. Suppose that the operating cost per charger is $c_j=2$ for all $j \in F$. The penalty for one passenger loss is $\beta_k=1$ for all $ k \in S$.

As established in Theorem \ref{thm:concavity}, the profit is a concave function of the number of chargers. This can be seen in Figure \ref{fig:profit_num_charger}. Moreover, we observe that the profit maximizing point is to have 3 chargers per station for the numerical example considered here.

\begin{figure}[!htb]
\centering
\scalebox{0.75}{
\includegraphics[width=\linewidth]{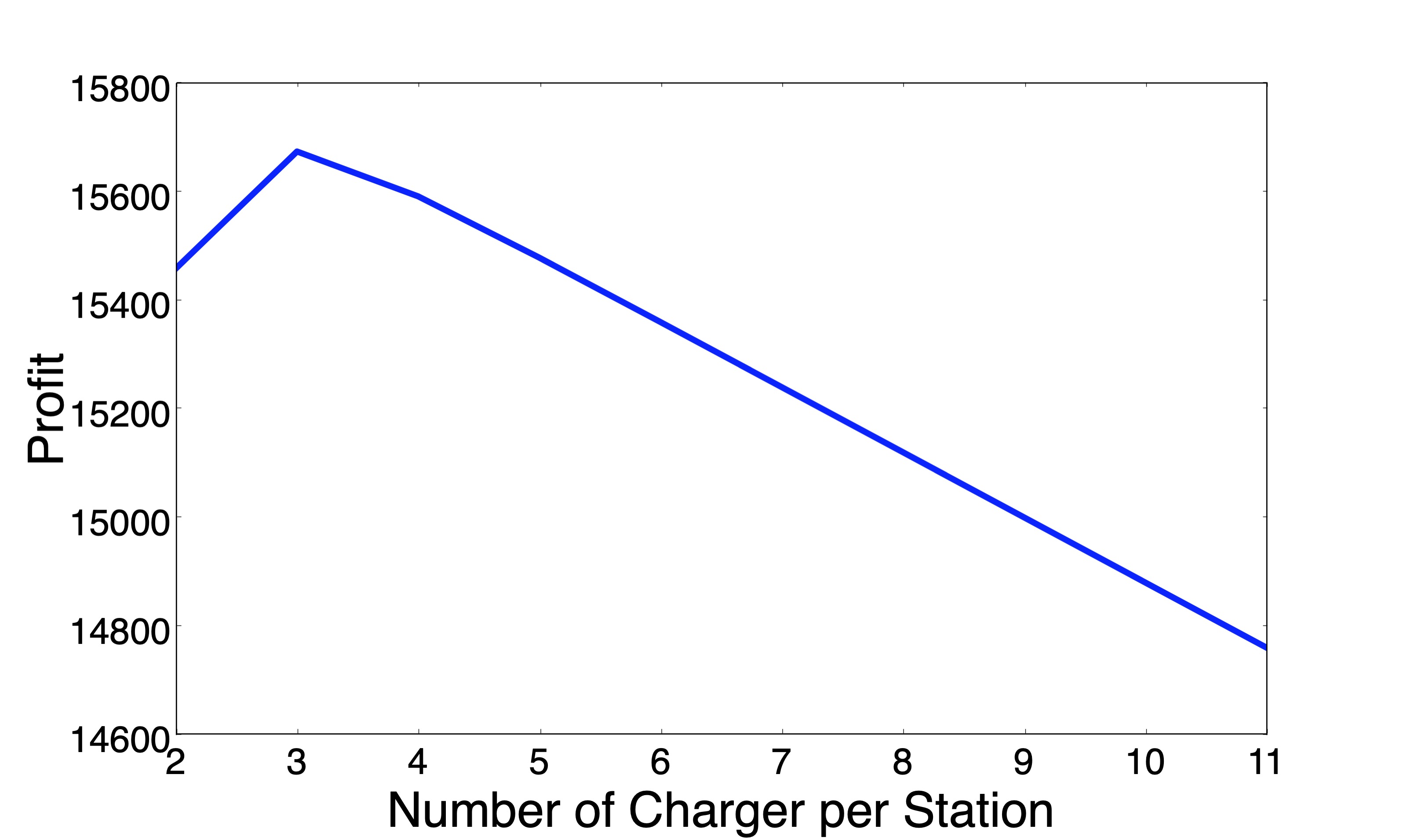}
}
\caption{Profit is a concave function of number of charger per station under charger operating costs and fixed fleet size.}
\label{fig:profit_num_charger}
\end{figure}

Figure \ref{fig:throu_num_charger} depicts the system throughput as a function of the number of chargers at each node. We observe that the system throughput increases as we increase the number of chargers. This is because the vehicles spend less time waiting for charging at charging points. When the number of chargers at each node is larger than a threshold, then the system throughput does not increase as we increase the number of chargers. For the current numerical example, this threshold is 4 charger per station.

\begin{figure}[!bth]
\centering
\scalebox{0.75}{
\includegraphics[width=\linewidth]{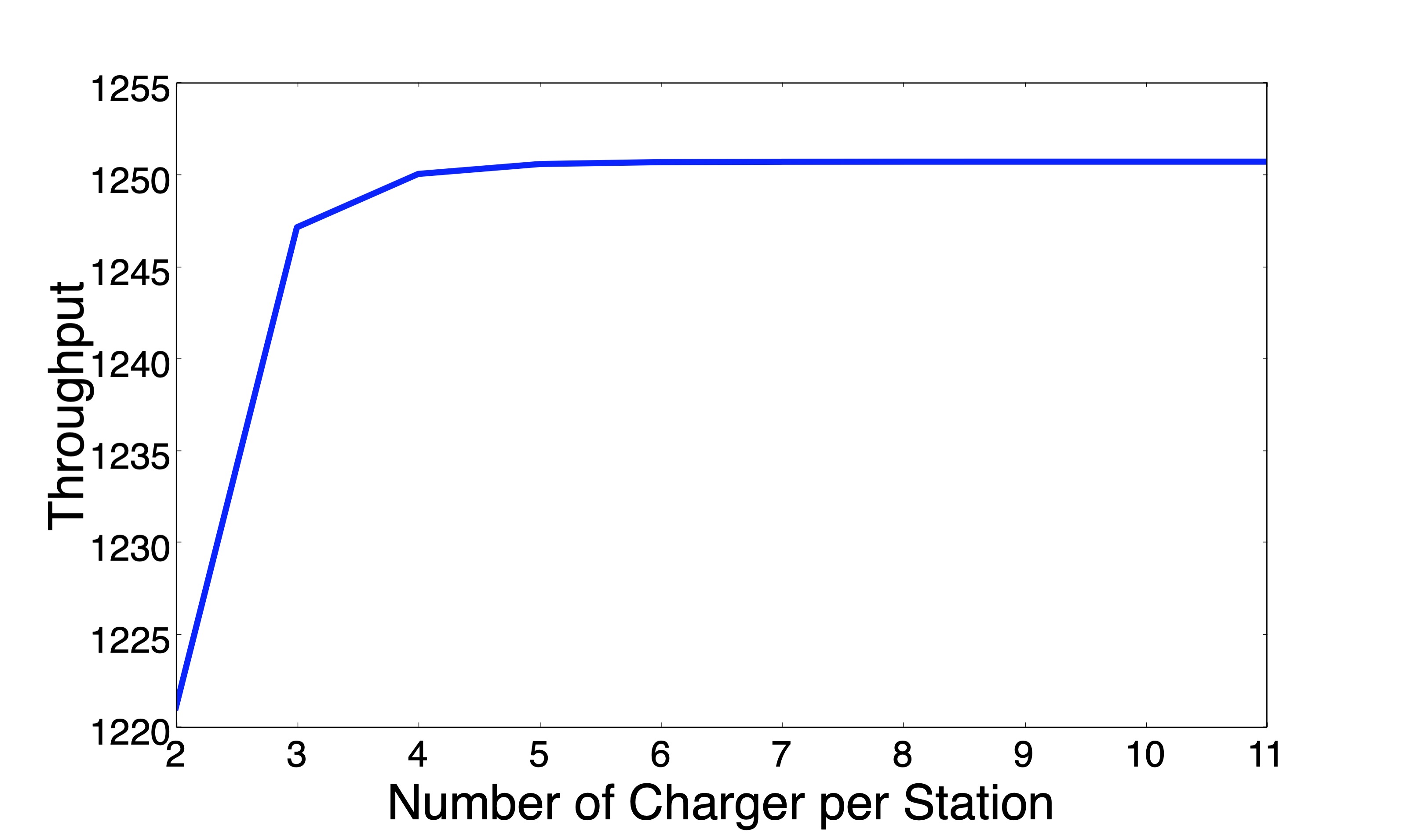}}
\caption{Throughput increases as the increase of number of charger.}
\label{fig:throu_num_charger}
\end{figure}

We now focus on the availability at one departure point as shown in Figure \ref{fig:avail_num_charger}. As more chargers become available, electric vehicles spend less time at charging points and more time at departure points. Therefore, the availability increases at single server queues.

\begin{figure}[!htb]
\centering
\scalebox{0.75}{
\includegraphics[width=\linewidth]{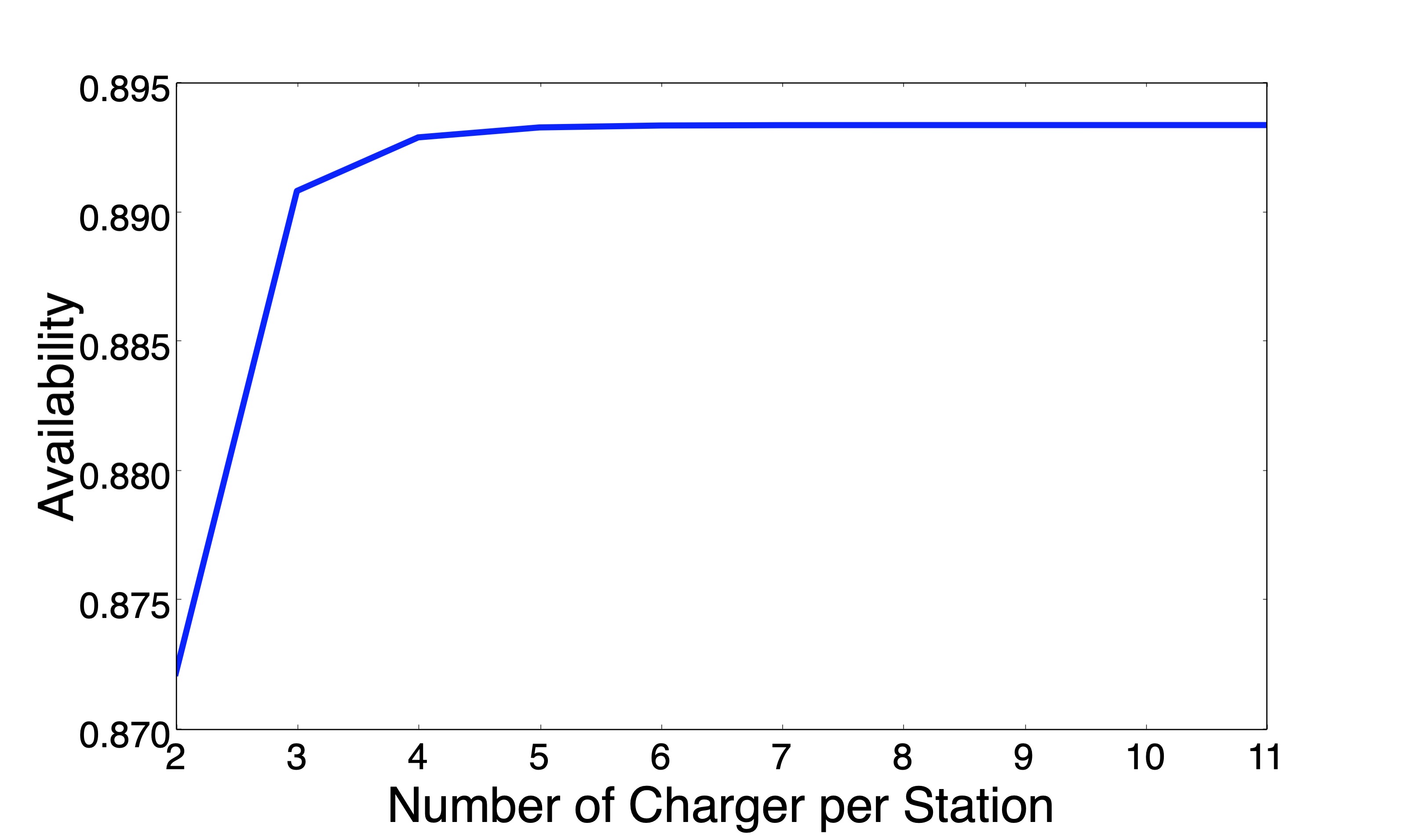}
}
\caption{Availability of vehicle increases if more chargers are installed at charging points}
\label{fig:avail_num_charger}
\end{figure}

Leveraging the Mean Value Analysis in Subsection \ref{sub:meanvalue}, the average queue length of each node (average number of vehicles at this node ) is depicted in Figure \ref{fig:qlength_num_charger}. As the number of chargers increases, the delay at charging points (Finite Server nodes) decreases because of less waiting time, and the vehicles move to the departure points (Single Server nodes), which increases their availability.

\begin{figure}[!htb]
\centering
\scalebox{0.75}{
\includegraphics[width=\linewidth]{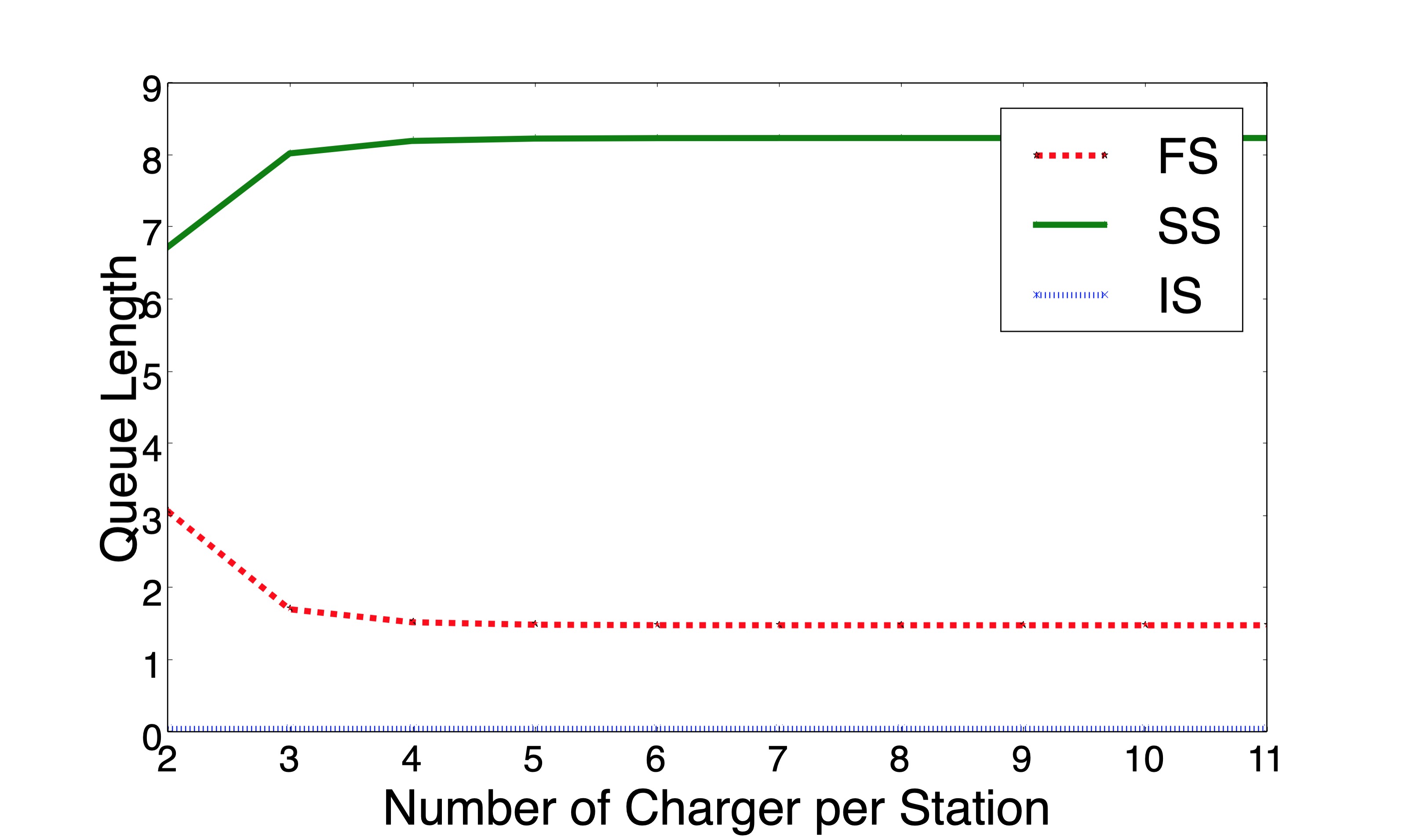}
}
\caption{Average number of vehicles at finite server nodes (FS:charging), single server nodes (SS:departure) and infinite server nodes (IS;travel) changes, with respect to number of chargers.}
\label{fig:qlength_num_charger}
\end{figure}

\subsection{Asymmetric Network Capturing the Relation between Downtown and Suburban Areas}

In this subsection, we want to show the charger allocation process when the network is not symmetric. We use a  three-station network in order to show it clearly and precisely, where the first station denotes the downtown while second and third station denotes suburban areas.

As shown in Figure \ref{fig:matrix}, we have the routing matrix $r_{ij}$ from \eqref{equ:r_ij}, where the orange part denotes that every vehicle will enter the departure point after finishing charging. The green part denotes the asymmetric transition probability $p_{ij}$ between stations, i.e., passengers departing from station 1 (downtown) will choose their destination equally between station 2 and 3 (suburb), while passengers departing from station 2 or 3 (suburb) will choose station 1 (downtown) as their destination with probability 60\% while station 3 or 2 (suburb) with probability 40\%. The blue part denotes that statistically one-third of EVs arriving at node $i \in F$ decide to be charged, while others choose to go to departure point directly without charging, i.e., $\bar{p_i}= \frac{1}{3}$, $\forall i \in F$. Other parts in this routing matrix are zero.
    
    \begin{figure}[!htb]
    \centering
    \scalebox{0.9}{
    \includegraphics[width=\linewidth]{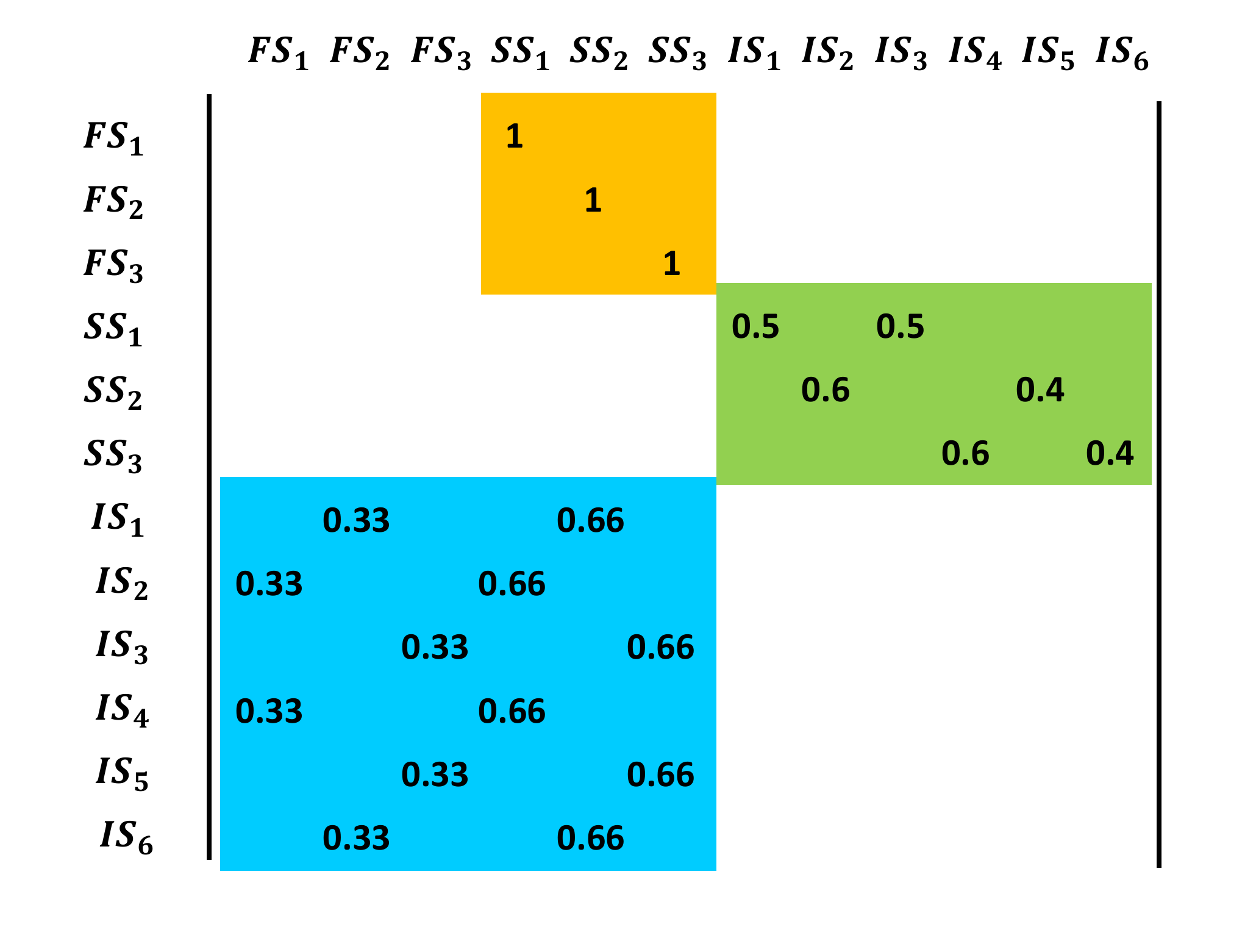}
    }
    \caption{Routing matrix $r_{ij}$ between nodes under asymmetric three-station network, where station 1 denotes downtown and station 2\&3 denote suburban areas.}
    \label{fig:matrix}
    \end{figure}
    
According to global balance \eqref{equ:lambda}, we have the relative throughput as follows, where the order of nodes is the same as the routing matrix.
    \begin{align}{\label{equ:visit_ratio}}
    \lambda =\Bigg[ \frac{3}{56},\frac{5}{112},\frac{5}{112},\frac{9}{56}, \frac{15}{112},\frac{15}{112}, \frac{9}{112},\frac{9}{112},\frac{9}{112},\frac{9}{112},\frac{3}{56},\frac{3}{56} \Bigg]
    \end{align}
    
\subsubsection{Charger Allocation} 
With a fixed number of vehicles $40$ in the system, we follow the Algorithm \ref{alg:charger} in order to find the optimal charger allocation in this network. 
    
    Distinguishing the difference of rent between downtown and suburb, we assume that the operating cost per charger is \$2 per hour at station 2 and 3 (suburb) while \$4 per hour at station 1 (downtown). We further assume the average revenue per service is \$30 and the penalty is \$ 1 per loss.
    
\begin{table}[h!]
\caption{Charger Allocation Algorithm}
\label{tab:charger}
\centering
\begin{tabular}{l l l l l l}
\hline
Step &$V$ & Profit & Revenue & Cost & Penalty \\
\hline
1 &(1,1,1) &458.16 &478.25 &8 &12.09\\
\space &(1,2,1) &457.53 &479.56 &10 &12.03\\
2 &(2,1,1) &533.58 &554.79 &12 &9.21\\
\space &(3,1,1) &530.05 &555.23 &16 &9.18\\
3 &(2,2,1) &553.46&575.89 &14 &8.43\\
\space &(3,2,1) &549.53 &575.93 &18 &8.40\\
4 &(2,2,2) &766.58 &783.21 &16 &0.63\\
\space &(2,3,2) &766.98 &785.55 &18 &0.57\\
5$^*$ &(3,2,2) &769.61 &790.00 &20 &0.39\\
\space &(3,3,2) &769.52 &791.85 &22 &0.33\\
\space &(4,2,2) &766.15 &790.51 &24 &0.36\\
% \hline
% Extra &(2,3,3) &767.21 &787.69 &20 &0.48\\
\hline
\end{tabular}
\label{Tab:Charger Allocation}
\end{table}

Following the Charger Allocation Algorithm (Algorithm \ref{alg:charger}), we can compute the allocation solution without listing all the candidates. As shown in table \ref{Tab:Charger Allocation},
the algorithm reduces the candidate size and computing complexity by leveraging the concavity property. After step 5, the algorithm terminates and we can claim that $(3,2,2)$ is the optimal solution in our setting, i.e., the number of chargers at (FS1,FS2,FS3) is $(3,2,2)$, without analyzing any other candidates.

\if 0
We further define a measure,  visit ratio per dollar operating cost (VRPD), to measure the efficiency of money spent on chargers with respect to visit ratios. For downtown (FS1), with visit ratio $\frac{3}{56}$ from \eqref{equ:visit_ratio} and operating cost \$4, its VRPD is $\frac{3}{224}$, while for suburban areas (FS2) with visit ratio $\frac{5}{112}$ and operating cost \$2, its VRPD is $\frac{5}{224}$. The result shows that, although the VRPD for suburban areas is high, relatively fewer chargers are put there, which counter-intuitively shows that the number of chargers may not increase as  VRPD, or VRPD is not a good measure for charger allocation.
\fi

If the constraint is active, i.e., $V$ is upper bounded by $\hat{V}$ as in \eqref{equ:upperbound}, this algorithm can still work. For example, if $\hat{V}=(2,5,5)$ is the upper bound of $V$, then the algorithm terminates at $(2,3,3)$, indicating this is an approximate solution to the optimization under this constraint.

\subsubsection{Convolution}    
Under the fleet size $M=40$ and $V=(3,2,2)$, we can use convolution algorithm in Subsection \ref{sub:convolution} and \eqref{equ:marginal} to find the distribution of electric vehicles,  through the marginal distribution of each node, i.e., the probability of $n_i$ vehicles at node $i \in \mathcal{N}$ in the closed queueing system. 

For example, the following Figures \ref{fig:marginal_SS2} and \ref{fig:marginal_FS2} present the marginal distribution in the departure and charging points of suburban areas. The first figure shows that there is no vehicle waiting at SS2 with a probability 18\%, which means that the availability at SS2 is 82\%. The second figure shows that there is no vehicle charging at FS2 with probability 18\%, which indicates that the utilization of charging infrastructure is high.

\begin{figure}[!htb]
\centering
\scalebox{0.75}{
  \includegraphics[width=\linewidth]{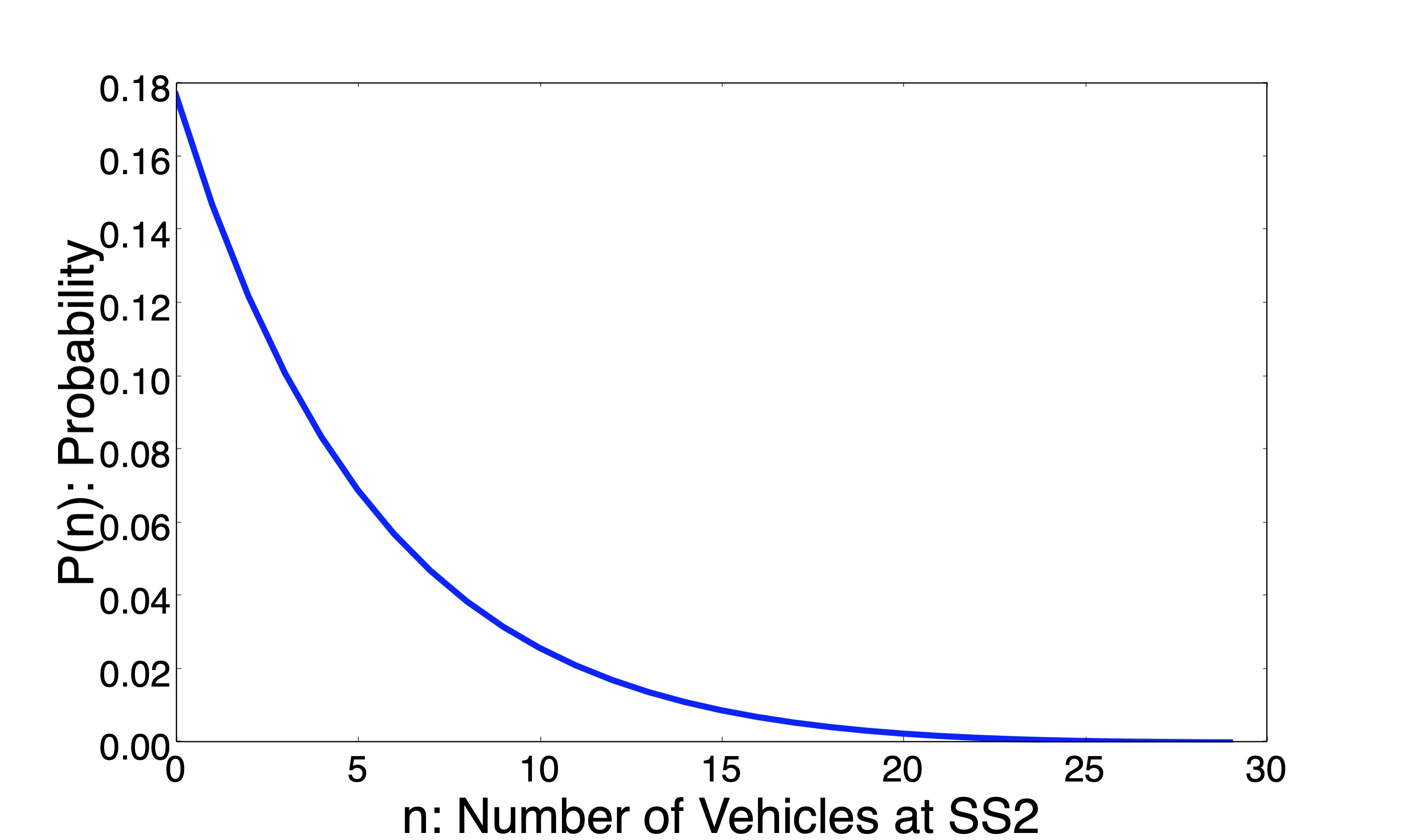}
}
\caption{Marginal Distribution of vehicles at SS2 (suburb-departure) under $M=40$ and $v=\{3,2,2\}$}
\label{fig:marginal_SS2}
\end{figure}

\begin{figure}[!htb]
\centering
\scalebox{0.75}{
\includegraphics[width=\linewidth]{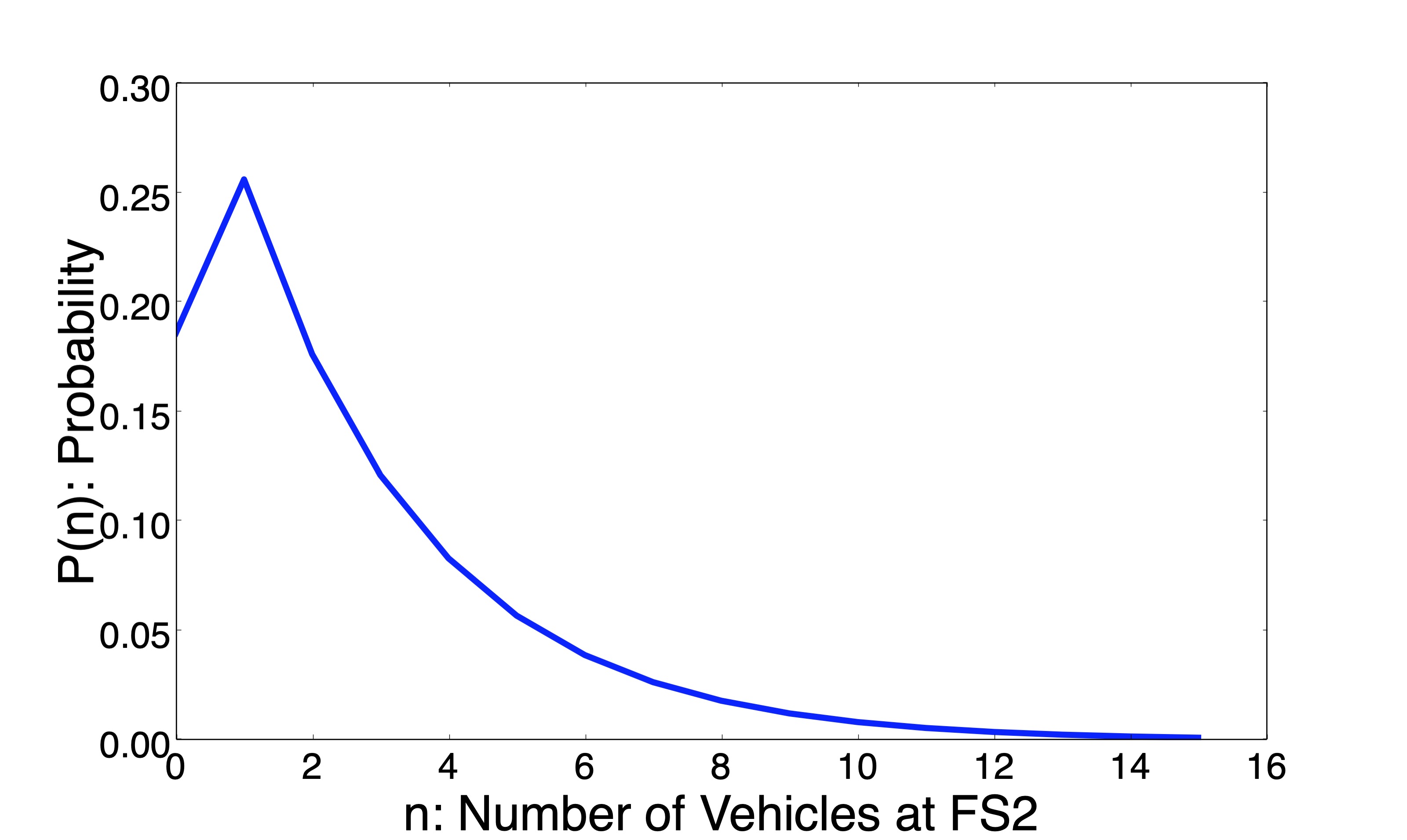}
}
\caption{Marginal Distribution of vehicles at FS2 (suburban charging station) under $M=40$ and $v=\{3,2,2\}$}
\label{fig:marginal_FS2}
\end{figure}

\subsection{One fast vs. Two slow} \label{sub:onefast}   
In this section, we show the result in Section \ref{sec:select} about one fast server vs. two slow servers. We test them in a simple closed queuing network of two queues, where the output of one queue is the input of another queue. The first queue is a single server queue with exponential service time, whose mean service time is $1/2$. The second queue has two choices, as shown in \ref{sec:select} with $t_0=1/2$. We compare the system throughput of two choices. 

Assume that there are only $10$ vehicles in this closed queueing network, we use Monte-Carlo simulation to find the system throughput under gamma distribution. 

\begin{figure}[!htb] 
\centering
\scalebox{0.75}{
\includegraphics[width=\linewidth]{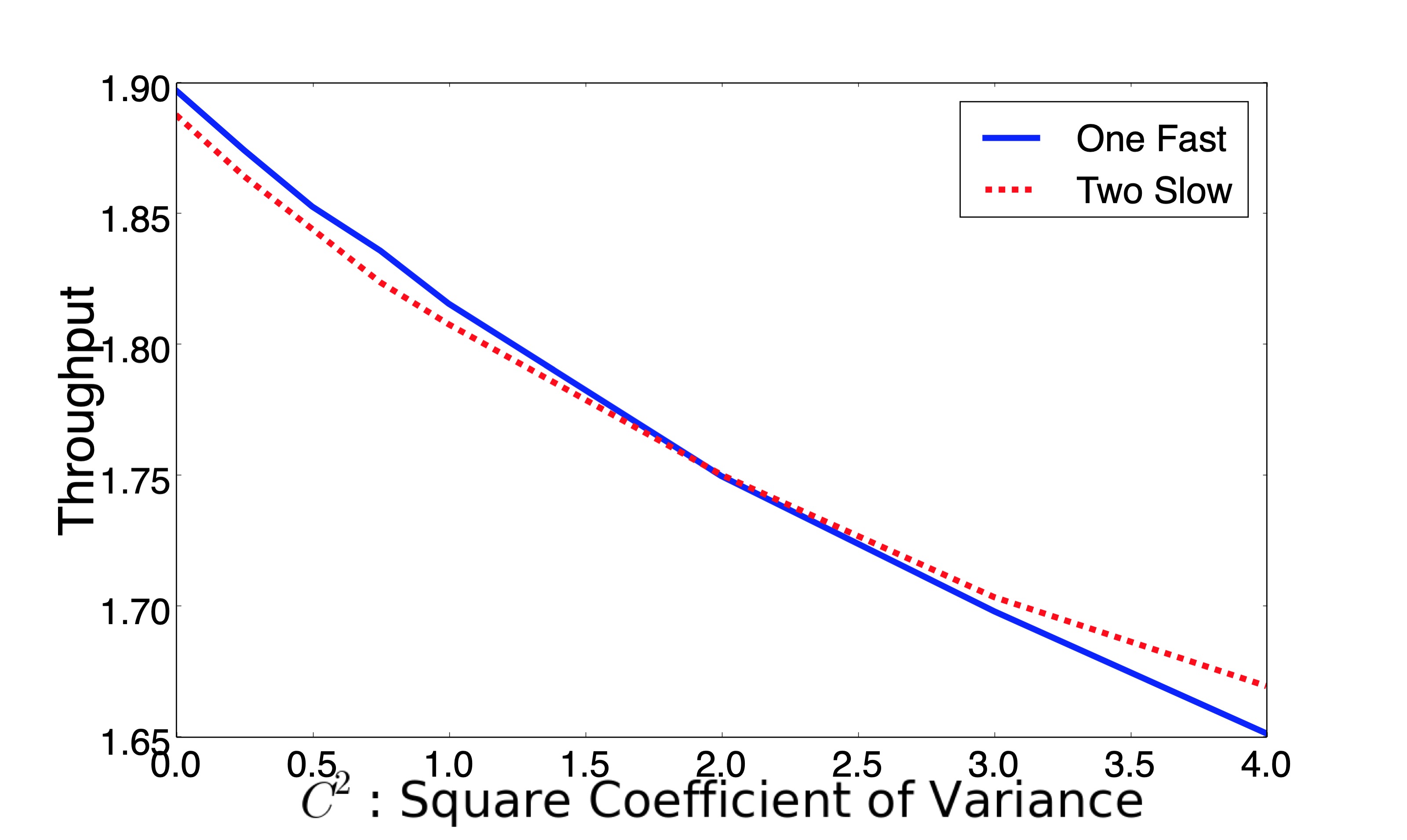}
}
\caption{Throughput comparison of one fast charger and two slow chargers under gamma distributed charging time in a closed queueing system.}
\label{fig:one_fast}
\end{figure}

As shown in Figure \ref{fig:one_fast}, when squared coefficient of variance $c^2$ is small, $D_1<D_2$ and thus one fast server has larger throughput. As $c^2$ increases, the gap between two choices closes. When $c^2$ becomes large, two slow servers outperform one fast server, i.e.,  $D_1 > D_2$ and the throughput under two slow servers is larger.     

If we keep increasing the number of chargers, the curve will be flatter. If the number of chargers exceeds the number of vehicles in the system, the throughput does not change with respect to the squared coefficient of variance, as shown in Lemma \ref{lemma:finite_insen}.

\begin{figure}[!htb] 
	\centering
	\scalebox{1}{
		\includegraphics[width=\linewidth]{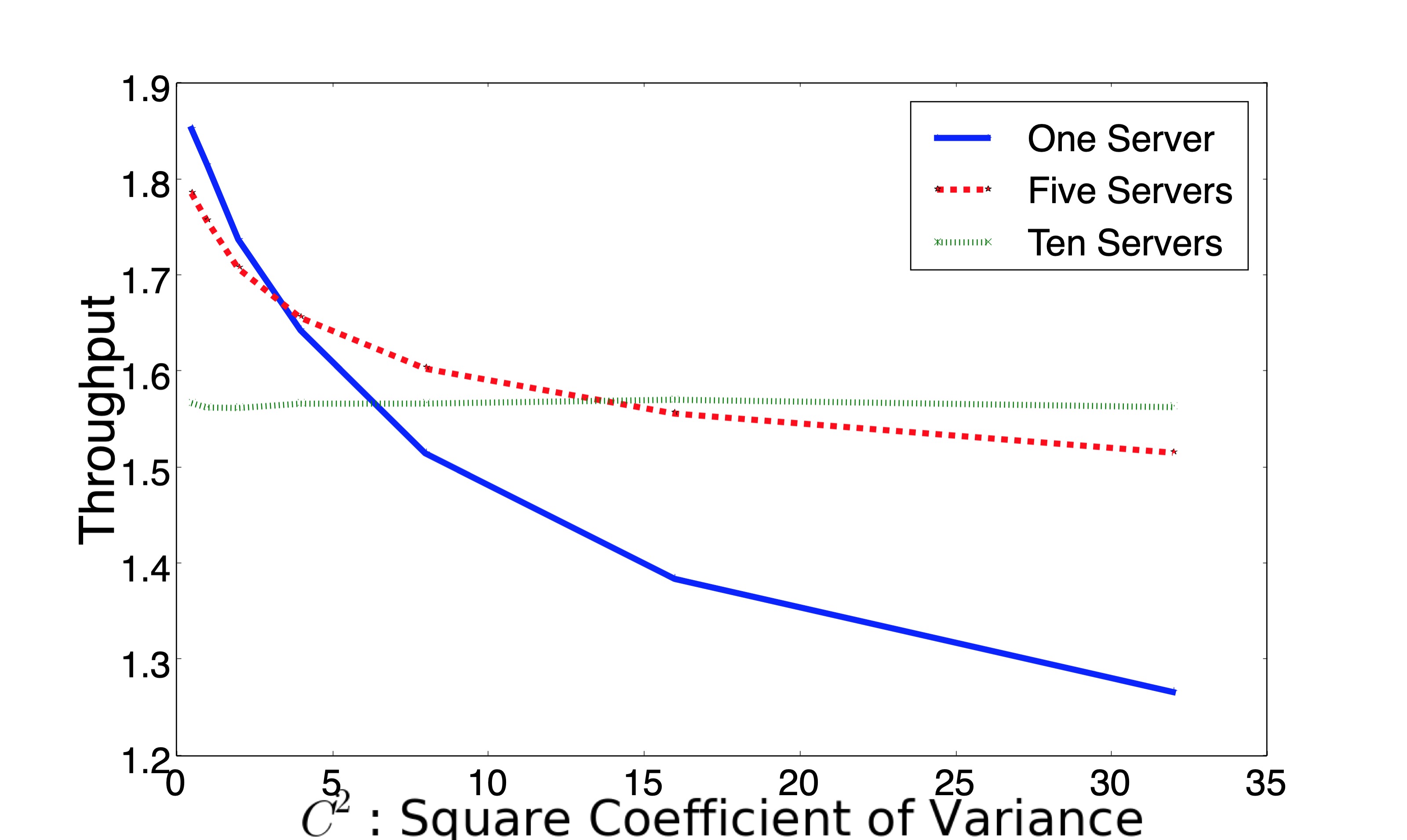}
	}
	\caption{Throughput comparison of one, five and ten servers under Gamma distributed charging time }
	\label{fig:Monte Carlo_gamma}
\end{figure}

As shown in Figure \ref{fig:Monte Carlo_gamma}, the threshold of $c^2$ where one server outperform five servers is around 4, which is larger than 1.9, the threshold where one server outperform two servers in Figure \ref{fig:one_fast}. When number of server is 10, which is the same as number of vehicles in the system, the throughput does not change with respect to the variance of charging time. Figure \ref{fig:Monte Carlo_IG} shows that similar results hold for Inverse Gausian distributed charging time and the threshold for $c^2$ is distribution dependent.  

\begin{figure}[!htb] 
	\centering
	\scalebox{1}{
		\includegraphics[width=\linewidth]{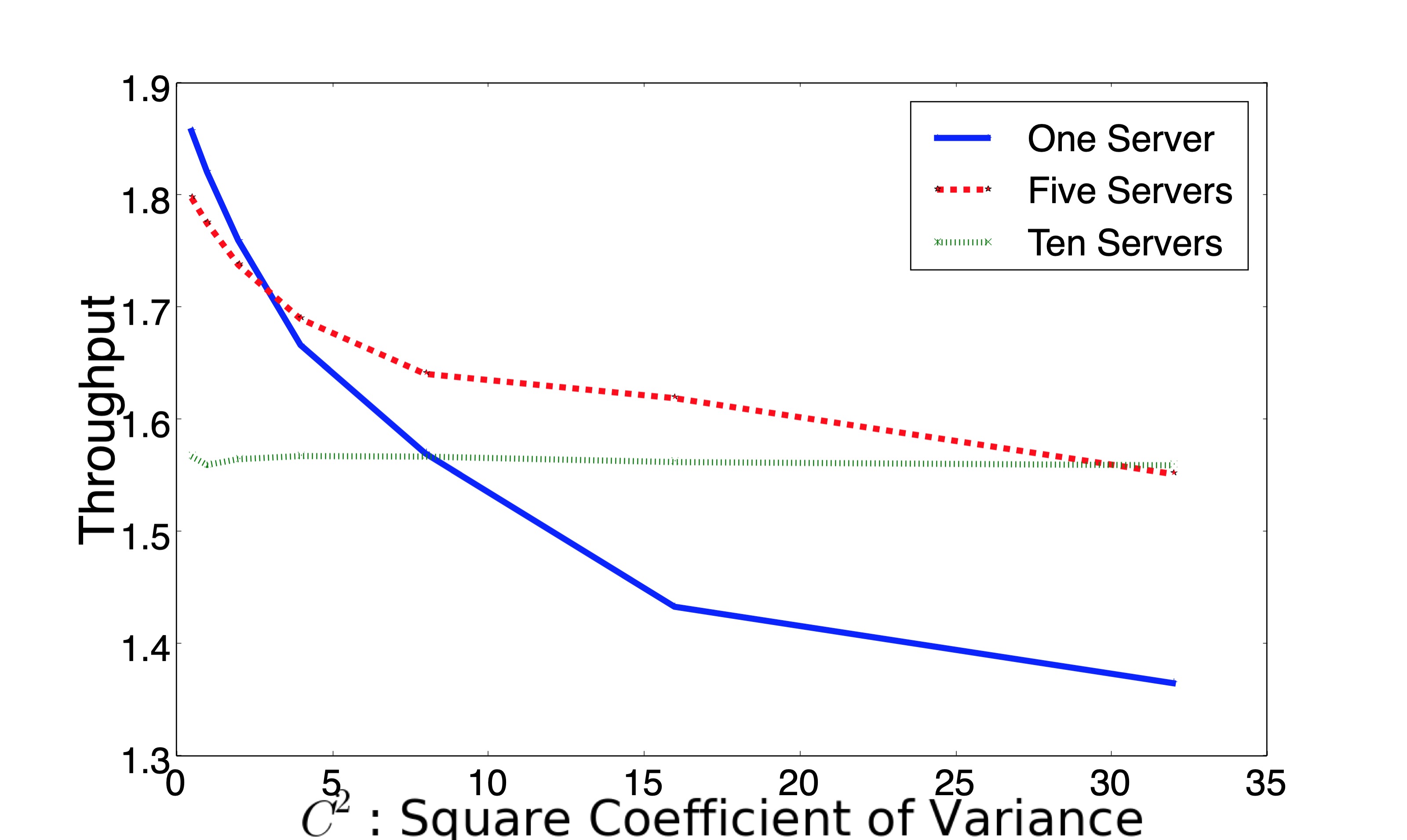}
	}
	\caption{Throughput comparison of one fast server and two slow servers under Inverse Gaussian distributed charging time }
	\label{fig:Monte Carlo_IG}
\end{figure}

\section{Conclusion}\label{sec:conclusion}
In this paper, we developed a closed queueing model for modeling a fleet of electric vehicles providing transportation service in a city. We considered the fleet sizing problem to maximize the profit of the system and the number of charges allocated within each charging station to maximize the total operational cost. We proved that the two problems lead to convex integer optimization problems. We developed a greedy algorithm for charger allocation and established its optimality if there are only two charging stations. When the variance of charging time becomes larger than the mean charging time, we showed using a stylized example that two slow chargers outperform one fast charger in terms of the total delay (waiting time plus the charging time). We further developed an approximation method for general passenger inter-arrival time distributions and the mean value analysis algorithm is provided for the performance analysis of the overall system.

Through this analysis, we gained many insights about fleet sizing, charger allocation, and charger selection. 
\begin{enumerate}
    \item As shown in Figure \ref{opti_m_num_charger}, the optimal fleet size can be reduced by adding more chargers. 
    \item Theorem \ref{thm:concavity} shows that adding chargers at any charging point will increase the system throughput and the availability of any departure point in the system.
    \item We posit that chargers should be allocated to the charging points that can bring high system throughput increment with a low cost, which is usually the areas with high visit ratios. This idea is inspired by the marginal allocation scheme studied in operations research \cite{fox1966discrete}. 
    \item  Fast chargers may be replaced by multiple slow chargers, if the standard variance of charging time is relatively large compared with mean charging time. We show this insight to be useful through a numerical simulation Subsection \ref{sub:onefast} (see Figure \ref{fig:one_fast}).
\end{enumerate}
Future research will address the rebalancing policy, more general charging time distributions and state-dependent routing strategies. 

\section*{Acknowledgment}
This research is supported by Ford Motor Company under the University Alliance Project.

\section*{References}
\bibliography{sample_electric}

\begin{thebibliography}{10}
\expandafter\ifx\csname url\endcsname\relax
  \def\url#1{\texttt{#1}}\fi
\expandafter\ifx\csname urlprefix\endcsname\relax\def\urlprefix{URL }\fi
\expandafter\ifx\csname href\endcsname\relax
  \def\href#1#2{#2} \def\path#1{#1}\fi

\bibitem{pavlenko2019does}
N.~Pavlenko, P.~Slowik, N.~Lutsey, When does electrifying shared mobility make
  economic sense?, Working Paper (2019).

\bibitem{weldon2018long}
P.~Weldon, P.~Morrissey, M.~O’Mahony, Long-term cost of ownership comparative
  analysis between electric vehicles and internal combustion engine vehicles,
  Sustainable Cities and Society 39 (2018) 578--591.

\bibitem{shen2019optimization}
Z.-J.~M. Shen, B.~Feng, C.~Mao, L.~Ran, Optimization models for electric
  vehicle service operations: A literature review, Transportation Research Part
  B: Methodological (2019).

\bibitem{hodgson1990flow}
M.~J. Hodgson, A flow-capturing location-allocation model, Geographical
  Analysis 22~(3) (1990) 270--279.

\bibitem{he2013optimal}
F.~He, D.~Wu, Y.~Yin, Y.~Guan, Optimal deployment of public charging stations
  for plug-in hybrid electric vehicles, Transportation Research Part B:
  Methodological 47 (2013) 87--101.

\bibitem{boyaci2015optimization}
B.~Boyac{\i}, K.~G. Zografos, N.~Geroliminis, An optimization framework for the
  development of efficient one-way car-sharing systems, European Journal of
  Operational Research 240~(3) (2015) 718--733.

\bibitem{weikl2015practice}
S.~Weikl, K.~Bogenberger, A practice-ready relocation model for free-floating
  carsharing systems with electric vehicles--mesoscopic approach and field
  trial results, Transportation Research Part C: Emerging Technologies 57
  (2015) 206--223.

\bibitem{he2017service}
L.~He, H.-Y. Mak, Y.~Rong, Z.-J.~M. Shen, Service region design for urban
  electric vehicle sharing systems, Manufacturing \& Service Operations
  Management 19~(2) (2017) 309--327.

\bibitem{george2011fleet}
D.~K. George, C.~H. Xia, Fleet-sizing and service availability for a vehicle
  rental system via closed queueing networks, European Journal of Operational
  Research 211~(1) (2011) 198--207.

\bibitem{fanti2014fleet}
M.~P. Fanti, A.~M. Mangini, G.~Pedroncelli, W.~Ukovich, Fleet sizing for
  electric car sharing system via closed queueing networks, in: 2014 IEEE
  International Conference on Systems, Man, and Cybernetics (SMC), IEEE, 2014,
  pp. 1324--1329.

\bibitem{zhang2016control}
R.~Zhang, M.~Pavone, Control of robotic mobility-on-demand systems: a
  queueing-theoretical perspective, The International Journal of Robotics
  Research 35~(1-3) (2016) 186--203.

\bibitem{zhang2018analysis}
R.~Zhang, F.~Rossi, M.~Pavone, Analysis, control, and evaluation of
  mobility-on-demand systems: a queueing-theoretical approach, IEEE
  Transactions on Control of Network Systems 6~(1) (2018) 115--126.

\bibitem{iglesias2019bcmp}
R.~Iglesias, F.~Rossi, R.~Zhang, M.~Pavone, A {BCMP} network approach to
  modeling and controlling autonomous mobility-on-demand systems, The
  International Journal of Robotics Research 38~(2-3) (2019) 357--374.

\bibitem{baskett1975open}
F.~Baskett, K.~M. Chandy, R.~R. Muntz, F.~G. Palacios, Open, closed, and mixed
  networks of queues with different classes of customers, J. ACM 22~(2) (1975)
  248--260.

\bibitem{gelenbe1998introduction}
E.~Gelenbe, G.~Pujolle, Introduction to queueing networks, Vol.~2, Wiley New
  York, 1998.

\bibitem{balsamo2000product}
S.~Balsamo, Product form queueing networks, in: Performance Evaluation: Origins
  and Directions, Springer, 2000, pp. 377--401.

\bibitem{balsamo2007queueing}
S.~Balsamo, A.~Marin, Queueing networks, in: International School on Formal
  Methods for the Design of Computer, Communication and Software Systems,
  Springer, 2007, pp. 34--82.

\bibitem{wolff1982poisson}
R.~W. Wolff, Poisson arrivals see time averages, Operations Research 30~(2)
  (1982) 223--231.

\bibitem{chandy1977product}
K.~M. Chandy, J.~H. Howard~Jr, D.~F. Towsley, Product form and local balance in
  queueing networks, Journal of the ACM (JACM) 24~(2) (1977) 250--263.

\bibitem{shanthikumar1988second}
J.~G. Shanthikumar, D.~D. Yao, Second-order properties of the throughput of a
  closed queueing network, Mathematics of Operations Research 13~(3) (1988)
  524--534.

\bibitem{shanthikumar1988server}
J.~G. Shanthikumar, D.~D. Yao, On server allocation in multiple center
  manufacturing systems, Operations Research 36~(2) (1988) 333--342.

\bibitem{fox1966discrete}
B.~Fox, Discrete optimization via marginal analysis, Management science 13~(3)
  (1966) 210--216.

\bibitem{smith2018introduction}
J.~M. Smith, Introduction to Queueing Networks: Theory \& Practice, Springer,
  2018.

\bibitem{buzen1973computational}
J.~P. Buzen, Computational algorithms for closed queueing networks with
  exponential servers, Communications of the ACM 16~(9) (1973) 527--531.

\bibitem{reiser1980mean}
M.~Reiser, S.~S. Lavenberg, Mean-value analysis of closed multichain queuing
  networks, Journal of the ACM (JACM) 27~(2) (1980) 313--322.

\bibitem{akyildiz1988mean}
I.~F. Akyildiz, G.~Bolch, Mean value analysis approximation for multiple server
  queueing networks, Performance Evaluation 8~(2) (1988) 77--91.

\bibitem{curry2010manufacturing}
G.~L. Curry, R.~M. Feldman, Manufacturing systems modeling and analysis,
  Springer Science \& Business Media, 2010.

\bibitem{serfozo2012introduction}
R.~Serfozo, Introduction to stochastic networks, Vol.~44, Springer Science \&
  Business Media, 2012.

\bibitem{lavenberg1983computer}
S.~Lavenberg, Computer performance modeling handbook, Elsevier, 1983.

\bibitem{shanthikumar1987optimal}
J.~G. Shanthikumar, D.~D. Yao, Optimal server allocation in a system of
  multi-server stations, Management Science 33~(9) (1987) 1173--1180.

\bibitem{topkis1998supermodularity}
D.~M. Topkis, Supermodularity and complementarity, Princeton university press,
  1998.

\end{thebibliography}

\appendix
\renewcommand*{\thesection}{\appendixname~\Alph{section}}

\section{Proof of Lemma \ref{lemma:BCMP}}
\begin{enumerate}
	\item For a product form closed queueing network with $M$ vehicles, equation \eqref{equ:Lambda2} and \eqref{equ:Lambda1} are given  by \cite[p.27]{serfozo2012introduction}   where $\lambda_i$ is the visit ratio defined in \eqref{equ:lambda} and $G(M)$ is the normalization factor defined in \eqref{equ:p_n} and \eqref{equ:gm}.
	
	\item 
	The idea comes from \cite{george2011fleet} and we rewrite the equations with our notations as follows. 
	The marginal distribution, i.e., the probability of $n_i$ vehicles at node $i \in \mathcal{N}$ is:
	\begin{align}{\label{equ:marginal}}
		p_i(n_i)=\frac{\lambda_i ^{n_i}}{\prod_{k=1}^{n_i} u_i(k)}\frac{G_i(M-n_i)}{G(M)}
	\end{align}
	where $G_i(M-n_i)$ is the normalizing constant when node $i$ is removed and only $M-n_i$ vehicles remains in the system by \cite[p.128.]{lavenberg1983computer}. 
	For nodes $i \in S$, the special case when the node is a single server node (SS), we can compute the probability without computing $G_i(M-n_i)$  
	We define the relative utilization for single server node $i \in S$ as
	\begin{align*}
		\gamma_i=\frac{\lambda_i}{u_i(1)}=\frac{\lambda_i}{\alpha_i} 
	\end{align*}
	Then the probability of $n_i$ vehicles at departure point $i \in S$ can be simplified by \cite[p.128]{lavenberg1983computer} as follows. 
	\begin{align*}
		p_i(n_i)=\frac{\gamma_i^{n_i} [G(M-n_i)-\gamma_i G(M-n_i-1)]}{G(M)}
	\end{align*}
	The availability is defined as the stationary state probability that node $i$ has at least one vehicle and has the expression as follows. 
	\begin{align}
		A_i(M)=1-p_i(0)= \gamma_i \Lambda(M)=\frac{\lambda_i}{\alpha_i} \Lambda(M) 
	\end{align}
\end{enumerate}{}

\section{Proof of Lemma \ref{lemma:concave}}\label{app:concave}
As shown in \cite[Theorem 1]{shanthikumar1988second}, in a closed Jackson network (exponential service time), the system throughput $\Lambda(M)$, is nondecreasing concave with job population $M$, if the service rate $u_i(n_i)$ is nondecreasing concave with local queue length $n_i$, $\forall i \in \mathcal{N}$. 

We first consider the situation that the service time of the infinite server (IS) follows exponential distribution in our problem, then our network falls into the Jackson network. In our setting, $u_i(n_i)$ defined in \eqref{equ:u_i} is constant (Single Server), linear (Infinite Server) and nondecreasing concave (Finite Server) with respect to $n_i$, i.e., all the service rates satisfy the nondecreasing concavity condition. Therefore, the system throughput  $\Lambda(M)$ defined in  \eqref{equ:Lambda1} is also non-decreasing concave with $M$.

Secondly, if we change the service distribution of infinite server queues (IS) into a general distribution with the same mean, by Lemma \ref{lemma:insen}, the throughput $\Lambda(M)$ does not change and the non-decreasing concave property remains.

Finally, applying \eqref{equ:Lambda1}, we have
\begin{align*}
	\sum_{i\in I} z_i\Lambda_i(M)=\Lambda(M) \sum_{i\in I} z_i \lambda_i
\end{align*}
where both $z_i$ and $\lambda_i$ are independent of $M$, which means the first part of objective function is concave. 

As $g(M)$ is convex, then the second part of $f(M)$ is concave. As a conclusion, $f(M)$ is concave.

\section{Proof of Theorem \ref{thm:concavity}}\label{app:concavity}
Let us first prove the first statement when the travel time distributions along the infinite server nodes are exponential. We then invoke Lemma \ref{lemma:insen} to conclude the first statement. This immediately yields the other two assertions.
\begin{enumerate}
    \item Firstly, assume that all the service time are exponentially distributed in our problem. From \eqref{equ:u_i}, we find that for all queues $i \in \mathcal{N}$, the service time $u_i(n_i)$ is increasing concave with $n_i$. From \cite[Theorem 1]{shanthikumar1987optimal}, in such closed queueing network, if there is a finite server node $j \in F$ with $v_j$ servers, i.e. service rate $u_j(n_j)=u_j(1) min\{n_j,v_j\}$, then the system throughput function $\Lambda(V)$ is increasing concave with $v_j$, i.e. $\Lambda(V)+\Lambda(V+2 e_j) \leq 2\Lambda(V+e_j)$. 
    
    Next, we change the service time distribution of infinite-server queues (IS) from exponential distribution to a general distribution without changing the mean. From insensitivity property in Lemma \ref{lemma:insen}, $\Lambda(V)$ for the general distribution travel time case remains the same as that of the exponentially distributed travel time. Therefore, we conclude that $\Lambda(V)+\Lambda(V+2 e_j) \leq 2 \Lambda(V+e_j)$ holds for closed queueing network in our setting. 
    \item As a result of Part 1 above, the first part of $h(V)$ is increasing concave with $v_j$, as $\bar{Z}$ is independent from $V$. Now, as $-\sum_{j \in F}c_j v_j $ is linear with $v_j$ and the concavity is preserved under addition, the second part is concave with $v_j$. Moreover, $-\sum_{k \in S}\beta_k \alpha_k$ does not depend on $V$. Therefore, $h(V)$ is concave with $v_j$.
    \item We now prove the second statement. From the first part of proof above, we know that $\Lambda(V)$ is increasing with $v_j$, $\forall j \in F$. Since $A_i(V) = \frac{\lambda_i}{\alpha_i} \Lambda(V)$, we conclude that $A_i(V)$ is increasing concave function in $v_j$ for all $ i, j \in F$.
\end{enumerate}

\section{Proof of Theorem \ref{thm:v2}}\label{app:v2}
    We first establish the supermodularity property \cite[p. 43]{topkis1998supermodularity} of the objective function below.
    \begin{lemma}{\label{lem:super}}
    	If $|F|=2$, $V=(v_1,v_2)$, the objective function is supermodular, i.e., $h(v_1+1,v_2-1)+h(v_1,v_2) \leq h(v_1+1,v_2)+h(v_1,v_2-1) $.
    \end{lemma}
    \begin{proof}{}
    	From \cite[Lemma 6 (ii)]{shanthikumar1988server}, the throughput is supper modular in the case of exponential service time, i.e.
    	\begin{align*}
    		\Lambda(v_1+1,v_2-1)+ \Lambda(v_1,v_2)
    		\leq \Lambda(v_1+1,v_2)+\Lambda(v_1,v_2-1)
    	\end{align*}{}
    	From Lemma \ref{lemma:insen}, it also hold for any general travel time distribution. 
    	Multiplied it with the constant $\bar{Z}$, and minus $c_1(2v_1+1)+c_2(2v_2-1)+2 \sum_{k \in S}\beta_k \alpha_k$ on both sides of the inequality, we can conclude that the supermodular property also holds for function $h(v_1,v_2)$.
    \end{proof}

	We now prove Theorem \ref{thm:v2} by transforming the problem and then using induction. Let $|V| = v_1+v_2$. Consider the following transformed optimization problem:
	\begin{align*}
	&\max_{V \in \Na^2} h(V) \quad\text{ such that } V \leq \hat{V}, |V| = n.
    \end{align*}
	We exploit the supermodularity of $h$, proved in Lemma \ref{lem:super}, and use induction on $n$ to establish that for every $n$, Algorithm \ref{alg:charger} outputs the optimal solution for the new constrained optimization problem above. 
	
	If $|V^*|=3$, which means there is only one extra charger needed to be allocated on $V^1=(1,1)$. As the algorithm evaluates both $V=(2,1)$ and $V=(1,2)$ allocations and there is no other allocation, Algorithm \ref{alg:charger} achieves the optimal solution.  
	
	Suppose the algorithm generates the optimal solution when $|V^*|=n$, denoted as $(a,n-a)$. For any $b<a$, we have 
	\begin{align}{\label{inequ:assump}}
		h(b,n-b) \leq h(a,n-a)
	\end{align}
	
	In order to prove the result, we only need to show that, when $|V^*|=n+1$, the optimal solution is either $(a+1, n-a)$ or $(a, n+1-a)$, which are evaluated in the algorithm. It suffices to show that any other allocation not evaluated by the algorithm $(b,n+1-b)$ cannot achieve a higher profit than either $(a+1, n-a)$ or $(a, n+1-a)$. We show the case with $b<a$, and the other case with $b>a+1$ can be proved similarly.
	
	As $h(v_1,v_2)$ is concave in $v_2$ from Theorem \ref{thm:concavity}, we conclude
	\begin{align}{\label{inequ:conca}}
		h(b,n+1-b)+h(b,n-a) \leq h(b,n-b) +h (b,n+1-a)
	\end{align}{}
	since $n-a <n+1-a \leq n-b <n+1-b$.	From Lemma \ref{lem:super}, $h(v_1,v_2)$ is supermodular. This implies
	\begin{align}{\label{inequ:supermo}}
		h(b,n+1-a)+h(a,n-a) \leq h(a,n+1-a) + h(b,n-a)
	\end{align}{}
	since $b<a$ and $n-a<n+1-a$. Adding up equation \eqref{inequ:assump}, \eqref{inequ:conca}, and \eqref{inequ:supermo}, we have 
	\begin{align}
		h(b,n+1-b) \leq h(a,n+1-a)
	\end{align}{}
	which concludes the result for case $b<a$ and finishes the induction.
	
	Finally, we discuss the constraint $v_1 \leq \hat{v_1}$ and $v_2 \leq \hat{v_2}$ in \eqref{equ:upperbound}. (i) If neither of the constraint is active, then the global maximum $V^*$ found in Algorithm \ref{alg:charger} is the optimal solution to the optimization problem. (ii) If only one constraint is activated, say $v_1=\hat{v_1}$, then
	the optimization becomes univariate and Algorithm \ref{alg:charger} generates the optimal solution, since the iteration continues as long as $h(v_1,v_2+1) > h(v_1,v_2)$. (iii) If both are active, then the upper bound $\hat{V}$ becomes the solution, from the last line in Algorithm \ref{alg:charger}. 
\end{document}